\documentclass[preprint,authoryear]{elsarticle}

\usepackage{amssymb}

\journal{}

\usepackage{graphicx}
\usepackage{amsmath,amssymb,amsfonts,latexsym}
\usepackage{color}
\usepackage{eucal}
\usepackage{xspace}
\usepackage{hyperref}

\usepackage{bm}

\usepackage[title]{appendix}

\newtheorem{proposition}{Proposition}
\newtheorem{proof}{Proof}

\usepackage[round]{natbib}

\usepackage{algorithm}
\usepackage[noend]{algpseudocode}

\DeclareMathOperator*{\argmax}{argmax} 

\begin{document}

\begin{frontmatter}
	
	
	
	\title{On an enhancement of RNA probing data using Information Theory
	}
	

	\author[add1]{Thomas J. X. Li}
\ead{jl9gx@virginia.edu}
	
	\author[add1,add2]{Christian M. Reidys\corref{mycorrespondingauthor}}
	\cortext[mycorrespondingauthor]{Corresponding author}
	\ead{duck@santafe.edu}
	
	\address[add1]{Biocomplexity Institute \& Initiative,
		University of Virginia, 995 Research Park Blvd,
		Charlottesville, VA, 22911}
	\address[add2]{	Department of Mathematics,
		University of Virginia,
	141 Cabell Dr, Charlottesville, VA 22903}
	
	\begin{abstract}
			Identifying the secondary structure of an RNA is crucial  for understanding its diverse regulatory functions.
		This paper focuses on how to enhance target identification
		in a Boltzmann ensemble of structures via  chemical probing data.
		We 	employ an information-theoretic approach  to solve the problem,
		via considering a variant of the R\'{e}nyi-Ulam game.
		Our framework is centered around the ensemble tree,
		a hierarchical bi-partition of the input ensemble,
		that is constructed by recursively querying about whether or not
		a base pair of maximum information entropy is contained in the target.
		These queries are answered via relating local with global probing data,
		employing the modularity in RNA secondary structures.
		We present that  leaves  of the tree are comprised of sub-samples 
		exhibiting a distinguished structure with high probability.
		In particular,	for a Boltzmann ensemble incorporating 
		probing data, which is well established in the literature, 
		the probability of our framework 
		correctly	identifying the target in the leaf  is greater than $90\%$.
		
	\end{abstract}
	
	
	
	
\end{frontmatter}



	
	\section{Introduction}
	
	
Computational methods for RNA secondary structure prediction 
have played an important role in unveiling the various   regulatory  functions of RNA.
In  the past four decades, these approaches 
have evolved from predicting a single minimum free energy (MFE) structure~\citep{Waterman:78s,Zuker:84}
to  Boltzmann  sampling an ensemble of possible structures~\citep{McCaskill,Ding:03}.
Despite its  success in a wide range of small RNAs, these
thermodynamics-based predictions are by no means perfect. 

In parallel, experiments by means
of chemical and enzymatic probing 
have become a frequently used  
technology to elucidate RNA
structure~\citep{Stern:88,Merino:05,Deigan:09}. 
These probing methods  use chemical reagents to bind unpaired nucleotides
and yield reactivities at nucleotide resolution.
To some extent, these reactivities
 provide information concerning  single-stranded  or double-stranded RNA
 regions. 
 Recent advances focus on
 the development of thermodynamics-based computational tools 
 that incorporate such experimental data~\citep{Deigan:09, Washietl:12,Zarringhalam:12}.

 While the use of probing data has significantly improved the prediction accuracy
 of \textit{in silico} structure prediction for several classes of RNAs~\citep{Lorenz:11},
 these methods have not solved the folding problem for large RNA systems, 
 such as long non-coding RNAs (lncRNAs, typically 200--20k bases).
 The reason is that the footprinting data does not identify base pairing partners
 of a given nucleotide. 
 In particular, probing data alone cannot distinguish short-range and long-range base
 pairings. 
 For long RNAs, the existence of the latter, however,
  has been shown experimentally~\citep{Lai:18} as well as
 theoretically~\citep{Li:18,Li:19}. Thus, 
 even combined with experimental data, 
 there are still numerous RNA folds consistent with the probing data.

	We assume that 
the ensemble of possible structures is in thermodynamic equilibrium, i.e. a Boltzmann ensemble.
For many classes of RNAs,  it is also reasonable to assume that the sequence folds into
 a unique structure,  the \emph{target}, which is contained in the ensemble. 
 Hence, the  problem of structure prediction gives rise  to the following challenge:\\
\begin{equation}\label{C}
\textit{How to enhance  target identification in a Boltzmann ensemble of structures?}
\end{equation}

 In this paper, we employ an information-theoretic approach in order to solve Problem~\ref{C},
via considering a variant of the R\'{e}nyi-Ulam game.
 Our framework is centered around the \emph{ensemble tree},
a hierarchical bi-partition of the input ensemble,
whose leaves are comprised of sub-samples 
exhibiting a distinguished structure with high probability.
Specifically,
the ensemble tree is constructed by recursively querying about whether or not
a base pair of maximum information entropy is contained in the target.
  We prove that the query of maximum entropy base pair 
splits the ensemble into two even parts and in addition provides maximum
reduction in the entropy of the ensemble. 
These questions can be answered in the affirmative, since the sequence is assumed 
to a single  target.
They are answered via relating additional probing data with the initial one,
employing the modularity in RNA secondary structures.
By this means,  we  identify the correct path
in the ensemble tree
 from the root to the leaf. 

The key result of this paper is that the probability of the ensemble tree 
correctly identifying the target in the  leaf is greater than $90\%$,
for the Boltzmann ensembles incorporating
probing data from sequences of length $300$, see Section~\ref{S:StaProperty}.
To demonstrate the result, we take into consideration three components.
Firstly,
we utilize a $q$-Boltzmann sampler with signature distance filtration,
which is well suited for  Boltzmann ensembles subjected to the probing data
constraint~\citep{Deigan:09, Zarringhalam:12}, see
Section~\ref{S:bolt}.
Secondly, we  consider the error rates arisen from answering the queries via probing data
in Section~\ref{S:modular}. 
We show that these error rates can be significantly reduced via repeated queries in Section~\ref{S:StaProperty}.
Thirdly, we prove that the leaf with low information entropy  contains a distinguished structure,
see Section~\ref{S:bolt2}.
We  present that, once in the correct leaf,
the probability  the distinguished structure being identical to the target is almost always correct.

To summarize,  the key points of our approach are:
	\begin{enumerate}
		\item our method starts with a Boltzmann sample and derives a sub-sample that contains the target with high probability,
		\item the derivation is facilitated by means of the ensemble tree, and the identification
		of the correct path from root to leaf,  is obtained by a variant of the R\'{e}nyi-Ulam game,
		\item  the answers to the respective queries are  inferred from  chemical probing,
		by relating additional probing data to initial one using modularity.
\end{enumerate}

This paper is organized as follows: in Section~\ref{S:search}, we
introduce the main elements of our framework: 
the R\'{e}nyi-Ulam game, the Boltzmann ensemble, base-pair queries
and  the ensemble tree.
In Section~\ref{S:answers}, we demonstrate how to integrate 
additional probing data with the initial ones allowing to
 answer the queries, thereby identifying the correct path.
In Section~\ref{S:properties}, we analyze the ensemble tree
and present that our approach 
 identifies the target reliably and efficiently.
Finally, we integrate and discuss our results  in Section~\ref{S:dis}.

\section{Some background}
\label{S:search}


\subsection{The R\'{e}nyi-Ulam game}

We now approach Problem~\ref{C} via the R\'{e}nyi-Ulam game, a two-person game, played
by a questioner (Q) and an oracle, (O). Initially O thinks of an integer, $Z$, between
one and one million and Q's objective is to identify $Z$, asking yes-no questions. O is
allowed to lie at a rate specific to yes and no, respectively.

The R\'{e}nyi-Ulam game has been  extensively studied since the early works by~\cite{Renyi:61,Ulam:76},
and has  various applications such as adaptive error-correcting codes in the context of noisy
communication~\citep{Shannon:48,Berlekamp:68}. Depending on the respective application scenario,
numerous variants of the R\'{e}nyi-Ulam game have been considered, specifying the format of admissible
queries or the way O lies~\citep{Pelc:89,Spencer:92}.

In what follows, we shall play the following version of the game:
O holds a set of \emph{bit strings} $y_1 y_2 \cdots y_l$ of finite length $l$,
not every bit string being equally likely selected and the queries have to following
format: ``Is the $i$th-bit of the bit string equal to $1$?'', i.e. ~Q executes \emph{bit query}. 
O's lies occur at \emph{random}, are \emph{independent} and \emph{context-dependent}.
Specifically, O lies with probability $e_0$ and $e_1$ in case of the truthful answer being ''No'' and ``Yes'',
respectively. The particular cases $e_0=0$  and $e_1=0$ have been studied in the context of \emph{half-lies}~\citep{Rivest:80}.

The majority of studies on the R\'{e}nyi-Ulam game to date is \emph{combinatorial}.
That is, they stipulated the number of lies (or half-lies) being \textit{a priori} known and  
focused on finding optimal searching strategies which uses a minimum number of queries to identify
the target in all cases~\citep{Rivest:80,Spencer:92}.

Within the framework of this paper, we study two distinctly different manifestations of the oracle.
The first is embodied as an indicator random variable, whose distribution is derived from a modularity
analysis on RNA MFE-structures, see Section~\ref{S:modular}, and the second recruits experimental data, see
Section~\ref{S:experiment}. In both manifestations, erroneous responses arise intrinsically at random: either as a
result of the distribution of the r.v.~or intrinsic errors of experimental data.

By construction, this rules out a unique winning strategy for Q: instead, we consider the
\emph{average fidelity} or accuracy to identify the target utilizing a \emph{sub-optimal} number of
queries.
We shall propose an entropy-based strategy: at any point a query is selected relative to the subset of
bit strings coinciding with the target in all previously identified positions, that maximizes the
uncertainty reduction on the subset, see Section~\ref{S:bolt4}.
 

\subsection{The Boltzmann ensemble}
\label{S:bolt}

At a given point in time, an RNA sequence, $\mathbf{x}$, assumes a fixed secondary structure, by
establishing base pairings. Over time, however, $\mathbf{x}$ assumes a plethora of RNA secondary
structures appearing at specific rates, see ~\ref{A:1} for details and context on RNA.
These exist in an equilibrium ensemble expressed by the partition function~\citep{McCaskill} of
$\mathbf{x}$.

More formally, the \emph{structure ensemble}, $\Omega$ of $\mathbf{x}$ is a discrete probability space 
over the set of all secondary structures, equipped with the probability $p(s)$ of $\mathbf{x}$ folding
into $s$. We shall assume that the ensemble of structures is in thermodynamic equilibrium, the
distribution of these structures being described as a Boltzmann distribution. 
The \emph{Boltzmann probability}, $p(s)$, of the structure $s$ is a function of the
\emph{free energy} $E(s)$ of the sequence $\mathbf{x}$ folding into 
$s$,
computed via the Turner energy model~\citep{Mathews:99,Mathews:04}, see ~\ref{A:energy} for details.
The Boltzmann probability $p(s)$ is expressed as the Boltzmann factor $\exp{(-E(s)/R T)}$,
normalized by the \emph{partition function}, $Z=\sum_{s\in \Omega} \exp{(-E(s)/R T)}$, i.e. 
\[
p(s) = \frac{\exp{(-E(s)/R T)}}{Z},
\]
where $R$ denotes the universal gas constant and $T$ is the absolute temperature. The Boltzmann
distribution facilitates the computation of the partition function $Z$ for each substructure.
The partition function algorithm~\citep{McCaskill} for secondary structures computes $Z$
and, in particular, the base pairing probabilities based on the free energies for each structure
within the structure ensemble $\Omega$.

Let $p_{i,j}$ denote the probability of a base pairing between nucleotides $i$ and $j$
in the ensemble $\Omega$. Clearly, $p_{i,j}$ can be computed as the sum of probabilities of
all secondary structures that contain $(i,j)$, that is,
\[
p_{i,j} = \sum_{s\in \Omega} p(s) \delta_{i,j}(s),
\]
where $\delta_{i,j}(s)$ denotes the occurrence of the base pair $(i,j)$ in $s$.

The thermodynamics-based partition function
 has been extended to 
incorporate  chemical probing data to generate a Boltzmann ensemble, $\Omega_{\text{probe}}$. 
These approaches~\citep{Deigan:09, Washietl:12,Zarringhalam:12} transform 
structure probing data into a pseudo  energy term,  $\Delta G(s) $,
which reflects how well the structure agrees with the
probing data.
The Turner free energy is  then evaluated by
adding the pseudo  energy term to the loop-based energy,
i.e., $E_{\text{probe}}(s)= E(s)+ \Delta G(s) $.
The corresponding equilibrium  ensemble, $\Omega_{\text{probe}}$, is 
distorted in favor of structures that are consistent with probing
data, see ~\ref{A:probe}.

In ~\ref{S:probedata},  we utilize the $0$-$1$ signature, which is suited for probing data,
 and  quantify 
 the discrepancy between the Boltzmann ensemble and the target
 via the signature distance $d_{sn}$. 
We present  that
the average distance for an unrestricted ensemble $\Omega$
is $0.21 n$,
while the distance for an ensemble  $\Omega_{\text{probe}}$ incorporating probing data is 
reduced to  $0.03n$. 
This motivates us to define a \emph{$q$-Boltzmann ensemble},
 $\Omega^q$, 
which consists of structures having signature distance to  the target $s$ at most $q  n$,
i.e., $\Omega^q=\{ s'| d_{\text{sn}}(s',s) \leq q n \}$.
In particular, we  present that the ensemble $\Omega_{\text{probe}}$ has
an average normalized  signature distance similar to a $q$-ensemble having $q=0.05$.
In this paper we discuss unrestricted and restricted Boltzmann ensembles,
$\Omega$ and $\Omega^q$.

We shall employ \emph{greyscale diagrams} in order to visualize a sample of secondary structures
by superimposing them in one diagram, visualizing the base pairing probabilities.
A greyscale diagram displays each base pair $(i,j)$ as an arc with greyscale $1-p_{i,j}$,
where greyscale $0$ represents black and $1$ represents white, see Fig.~\ref{F:grey}. 	

\begin{figure}
	\centering
	\includegraphics[width=0.6\textwidth]{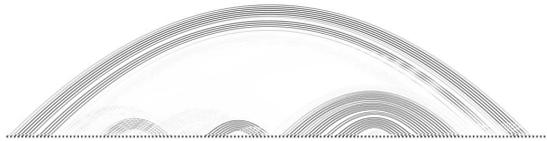}
	\caption
	{\small The greyscale diagram of of $1024$ Boltzmann sampled structures
	  of a random RNA sequence via \textsf{ViennaRNA}~\citep{Lorenz:11}.		
	}
	\label{F:grey}
\end{figure}

Instead of computing the entire ensemble, we shall consider sub-samples $\Omega'$ consisting
of $N$ secondary structures with multiplicities of $\mathbf{x}$ and refer to $\Omega'$ as the
\emph{sample}. For sufficiently large $N$ (typically of around size $10^3$, see~\cite{Ding:03}),
$\Omega'$ provides a good approximation of the Boltzmann ensemble $\Omega$.

A sample $\Omega'$ is a multiset of cardinality $N$ and for each structure $s$ in
$\Omega'$, its \emph{multiplicity}, $f(s)$, counts the frequency of $s$ appearing
in  $\Omega'$. Thus in the context of $\Omega'$, $p(s)$ is given by the $s$-multiplicity
divided by $N$, $p(s)=f(s)/N$. The base pairing probability $p_{i,j}$ has its
$\Omega'$-analogue $f(i,j)/N$, where $f(i,j)$ denotes the frequency of the
base pair $(i,j)$ appearing in  $\Omega'$. We shall develop our framework in the
context of the structure ensemble $\Omega$, and only reference the sample $\Omega'$,
in case the results are particular to $\Omega'$.


\subsection{The bit queries}
\label{S:bolt3}

Any structure over $n$ nucleotides is considered as a bit string of dimension
$\binom{n}{2}$, stipulating (1) a structure is completely determined by the set
of base pairs it contains and (2) any position can pair with any other position,
except of itself. 

The bit query now determines a single bit,
	i.e. whether or not the base pair $(i,j)$ is present in the
target, stipulating that a unique target is assumed by the sequence in question.
The target is  also assumed to have the Boltzmann probability 
as it appears in the ensemble.
We therefore associate
the  query about the target with 
a random variable, $X_{i,j}$, defined on the ensemble,
via
questioning the presence of  $(i,j)$ in each structure.
By construction, the distribution of $X_{i,j}$ is given by the base pairing probability 
$\mathbb{P}(X_{i,j}(s)=1)=p_{i,j}$.

Any base pair, $(i,j)$, has an \emph{entropy}, defined by the information entropy of
$X_{i,j}$, i.e. 
\[
H(X_{i,j})=-p_{i,j}  \log_2 p_{i,j}  -(1-p_{i,j} )  \log_2 (1-p_{i,j} ),
\]
where the units of $H$ are in bits. The entropy $H(X_{i,j}) $ measures the uncertainty of  the base
pair $(i,j)$ in $\Omega$ . When a base pair $(i,j)$  is certain to either exist or not, its entropy
$H(X_{i,j})$ is $0$. However, in case $p_{i,j} $ is closer to $1/2$, $H(X_{i,j})$ becomes larger.

The r.v. $X_{i,j}$ partitions the space $\Omega$ into two disjoint sub-spaces $\Omega_0$ and
$\Omega_1$, where  $ \Omega_k=\{s\in  \Omega :X_{i,j}(s)=k\}$ ($k=0,1$), 
and the induced distributions are given by
\[
p_0(s)=\frac{p(s)}{1-p_{i,j}} \quad \text{ for } s\in  \Omega_0,\qquad
p_1(s)=\frac{p(s)}{p_{i,j}} \quad \text{ for } s\in  \Omega_1.
\]

Intuitively, $H(X_{i,j}) $ quantifies the average bits of information
we would expect to gain about the ensemble when querying a   base pair $(i,j)$. 
This motivates  us to consider the \emph{maximum entropy base pairs}, the base pair $(i_0,j_0)$
having  maximum entropy among all base pairs in $\Omega$, i.e. 
\[
(i_0,j_0)= \argmax_{(i,j)}  H(X_{i,j}).
\]
As we shall prove in Section~\ref{S:bolt2}, $X_{i_0,j_0}$ produces maximally balanced
splits.

\subsection{The ensemble tree}
\label{S:bolt4}

Equipped with the notion of ensemble and bit query (i.e.~the respective maximum entropy base pairs),
we proceed by describing our strategy to identify the target structure as specified in Problem~\ref{C}.
The first step consists in having a closer look at the space of ensemble reductions.

Each split obtained by partitioning the ensemble $\Omega$ using r.v. $X_{i,j}$, can in turn be
bipartitioned itself via any of its maximum entropy base pairs. This recursive splitting induces the
{\it ensemble  tree}, $T(\Omega)$, whose vertices are sub-samples and in which its $k$-th layer represents
a partition of the original ensemble into $2^k$ blocks. $T(\Omega)$, is a rooted binary tree, in which
each branch represents a $X_{i,j}$-induced split of the parent into its two children.

Formally the process halts if either the resulting sub-spaces are all \emph{homogeneous}, i.e.~their
structural entropy is $0$, which means that they contain only copies of one structure, or it reaches a
predefined  maximum level $L$.
In our case we set the maximum level to be  $L=11$, that is, the height of the ensemble tree is at
most $10$. The procedure is described as follows:
\begin{enumerate}
	\item start with the ensemble $\Omega$.
	\item for each space $\Omega_{\mathbf{k}}$ with $H(\Omega_{\mathbf{k}})>0$, where
          $\mathbf{k}$ is a sequence of $0$s and $1$s 
	having length at most $L-1=10$, compute:
\begin{itemize}
\item 	select  the maximum entropy  base pair $X_{i_{\mathbf{k}},j_{\mathbf{k}}}$  of
  $\Omega_{\mathbf{k}}$ as the feature, i.e. 
	\[
	X_{i_{\mathbf{k}},j_{\mathbf{k}}}= \argmax_{(i,j)  \text{ in } \Omega_{\mathbf{k}}}  H(X_{i,j}).
	\] 
	\item 	split  $\Omega_{\mathbf{k}}$ into  sub-spaces $\Omega_{\mathbf{k}0}$ and
	$\Omega_{\mathbf{k}1}$ using the feature $X_{i_{\mathbf{k}},j_{\mathbf{k}}}$, that is,
$ \Omega_{\mathbf{k}l }=\{s\in  \Omega_{\mathbf{k}} :X_{i_{\mathbf{k}},j_{\mathbf{k}}}(s)=l\}$ for $l=0,1$, 
\end{itemize}	
	\item repeat Step $2$ until all new sub-spaces either have structural entropy $0$
	or reach the maximum level $11$.
\end{enumerate}

\begin{algorithm}
	\caption{Ensemble Tree}\label{ensemble}
	\begin{algorithmic}[1]
		\Procedure{$T$}{$\Omega$}
		\State $\mathbf{k} \gets \{\} $
		\State $\Omega_{\mathbf{k}} \gets \Omega$ 
		\Repeat
		\State $X_{i_{\mathbf{k}},j_{\mathbf{k}}}\gets \argmax_{(i,j)  \text{ in } \Omega_{\mathbf{k}}}  H(X_{i,j})$
		\State  $ \Omega_{\mathbf{k}l }\gets \{s\in  \Omega_{\mathbf{k}} :X_{i_{\mathbf{k}},j_{\mathbf{k}}}(s)=l\}$ for $l=0,1$
		\State Append $0$ or $1$ to $\mathbf{k}$
		\Until{$H(\Omega_{\mathbf{k}})=0$ or $|\mathbf{k}|=10$ }
		\State \textbf{return} $\{\Omega_{\mathbf{k}}\} $
		\EndProcedure
	\end{algorithmic}
\end{algorithm}

In Fig.~\ref{F:ensembletree} we display an ensemble tree.

\begin{figure}
	\centering
	\includegraphics[width=0.9\textwidth]{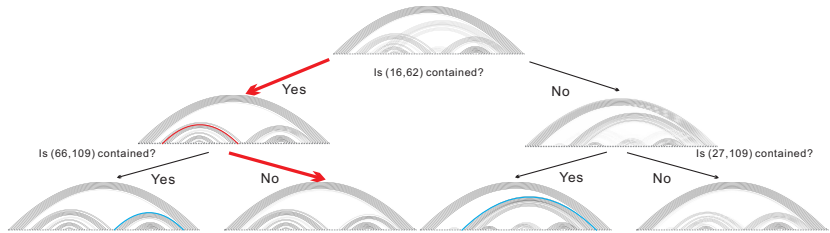}
	\caption
	{\small An ensemble tree having maximum level $3$.
		A  path
		from the root to a leaf  is identified in color red.
	}
	\label{F:ensembletree}
\end{figure}

\section{Path identification}
\label{S:answers}

Given the ensemble tree, we shall construct a path recursively starting
from the root to identify the leaf that contains the target. We shall do
so by successive bit queries about maximum entropy base pairs, see Fig.~\ref{F:ensembletree}.

As mentioned before, we employ two manifestations of the oracle, one using modularity based on
RNA-folding, see Section~\ref{S:modular}, and the other determining the existence of base pairs by
experimental means (Section~\ref{S:experiment}).

\subsection{The oracle via modularity and RNA folding }
\label{S:modular}

Here we shall employ \emph{modularity} of RNA structures, i.e. ~the loops,
which constitute the additive building blocks for the free energy have only marginal
dependencies. This can intuitively be understood by observing that any two loops
can only intersect in at most two nucleotides, see ~\ref{A:energy}.

Let us introduce the notion of embedding and extraction of a contiguous subsequence
or fragment
\begin{align*}
\epsilon_{i,j}((x_1,\dots, x_j),(y_1,\dots,y_m)) & =  (y_1,\dots,y_{i-1},x_1,\dots,x_j,y_{i+1},\dots y_m)\\
\xi_{i,j}(x_1,\dots,x_{n}) & =  ((x_i,\dots,x_{j}),(x_1,\dots,x_{i-1},x_{j+1},\dots,x_{n})).
\end{align*}
By construction, we have $\epsilon_{i,j}\circ \xi_{i,j}=\text{\rm id}$ and a contiguous subsequence or
fragment of an RNA sequence is called \emph{modular} if it being extracted folds into the same
arc configuration as it does embedded in the sequence. 

Next we show how to employ probing data to reliably answer whether or not a particular
(maximum entropy) arc is contained in the target structure. Structural modularity implies
that if this arc can indeed be found in the target structure, then a comparative analysis
of the probing data of the entire sequence with those of the extracted sequence, as well
as the remainder, concatenated at the cut points will exhibit distinctive similarity.
Modularity is a decisive discriminant, if, in contrast, random fragments do not exhibit such
similarity.

To quantify to what extent modularity can discriminate base pairs, 
we perform computational experiments on random sequences via splittings. For each sequence, we consider
its MFE structure $s$ computed via \textsf{ViennaRNA}~\citep{Lorenz:11}.
Given two positions $i$ and $j$, we cut the entire sequence $\mathbf{x}$  into two fragments,
$\mathbf{x}_{i,j}$ and the remainder  $\bar{\mathbf{x}}_{i,j}$, 
	i.e., $ \xi_{i,j}(\mathbf{x})=(\mathbf{x}_{i,j},\bar{\mathbf{x}}_{i,j})$.
Subsequently, the two fragments $\mathbf{x}_{i,j}$ and  $\bar{\mathbf{x}}_{i,j}$ refold
into their MFE structures $s_{i,j}$  and   $\bar{s}_{i,j}$, respectively,
which are combined into a structure $\epsilon_{i,j}(s_{i,j}, \bar{s}_{i,j})$.
If bases $i$ and $j$ are paired in $s$,  such a splitting is referred to as \emph{modular}
and the resulting structure is denoted by $s'$.
Otherwise, it is called \emph{random}, with  the output structure $s''$. 
We proceed by computing the base-pair and signature distance from the MFE $s$ to the structures
$s'$ or $s''$. The base-pair distance is one of the most frequently used metrics to quantify the similarity
of two different structures viewed as bit strings~\cite{Zuker:89,Agius:10}, the signature distance measures the similarity
between their signatures, which is well suited within the context of the probing profiles, see ~\ref{A:1}.

Fig.~\ref{F:shapeM} (LHS) compares the distribution of the signature distances $d_{\text{sn}}(s,s')$ and
$d_{\text{sn}}(s,s'')$ obtained from modular and random splittings, respectively. The structures induced by
modular splitting have much more similar probing signatures to their MFE structures, than those induced by
random splitting. The situation is analogous for base-pair distances, see Fig.~\ref{F:shapeM} (RHS).
Since these distances measure structural similarity, the data also indicates that, when $i$ and $j$ form a
base pair, the fragment $\mathbf{x}_{i,j}$ is more likely to fold into the same configuration as it does being
embedded, i.e.  $\mathbf{x}_{i,j}$ is modular.

\begin{figure}
	
	\begin{tabular}{cc}
		\includegraphics[width=0.45\textwidth]{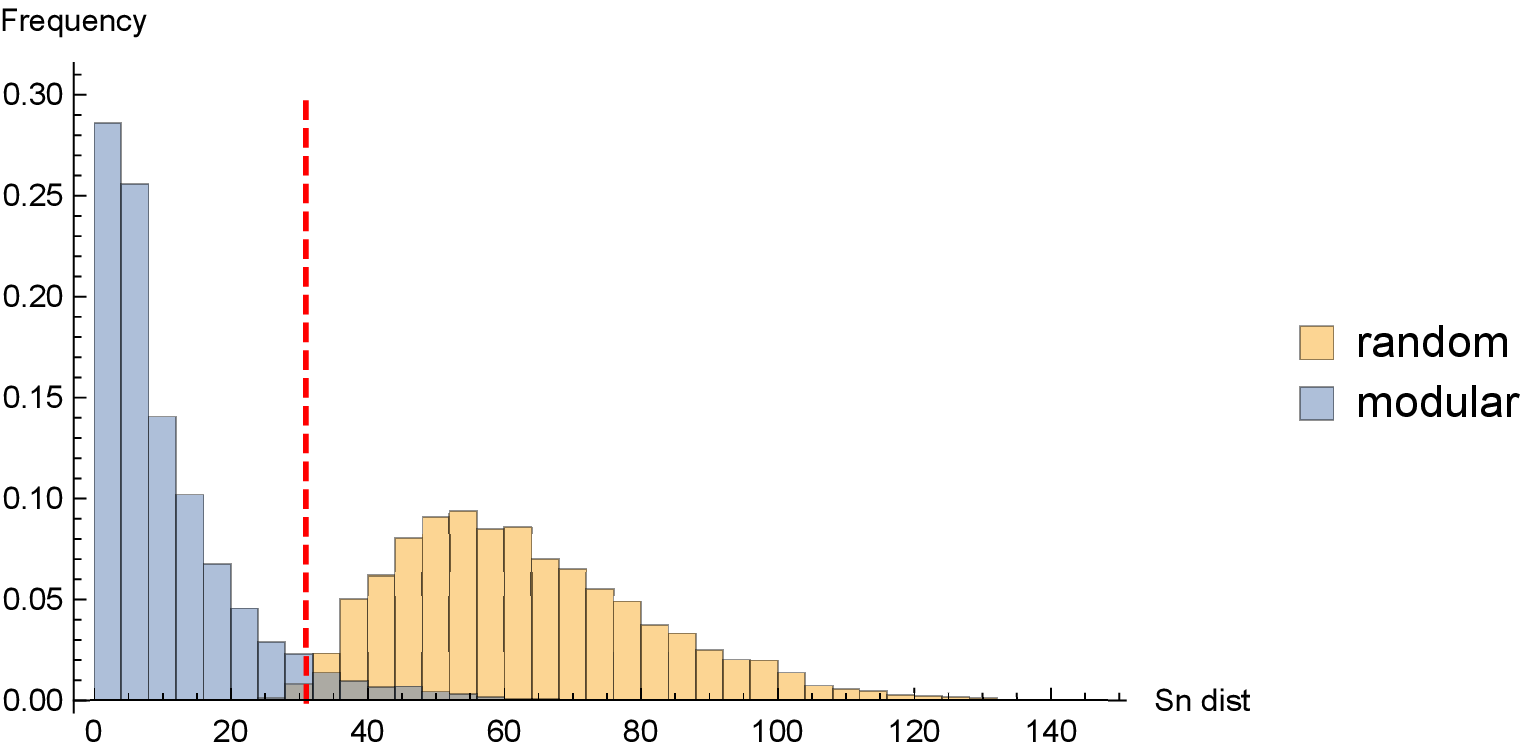}&
		\includegraphics[width=0.45\textwidth]{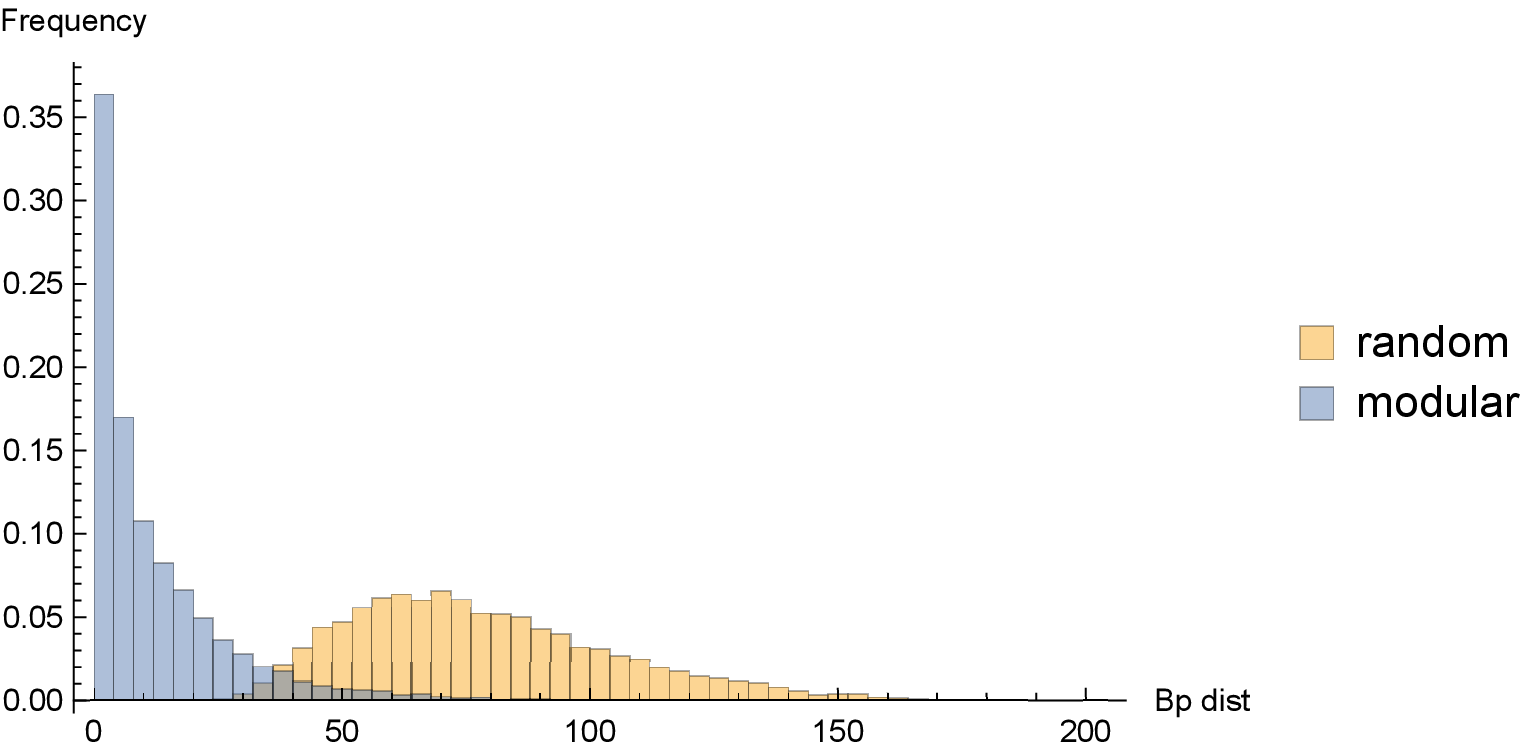}
		\\
	\end{tabular}
	\caption
	    {\small The distributions of the signature distances (LHS) $d_{\text{sn}}(s,s'),d_{\text{sn}}(s,s'')$
              and the base-pair distances (RHS) $d_{\text{bp}}(s,s'),d_{\text{bp}}(s,s'')$ 
		obtained from modular splitting (blue) and random splitting (orange). 
		We generated $8000$  random sequences $\mathbf{x}$ of length $500$ and computed their 
		structures $s$ (MFE), $s'$ (modular), $s''$ (random) via \textsf{ViennaRNA}~\citep{Lorenz:11}.
		The red dashed line (left) denotes ``threshold distance'', $31$ (see main text).
	}
	\label{F:shapeM}
\end{figure}

The data displayed in Fig.~\ref{F:shapeM} suggests the threshold distance, $\theta$, for signatures, by 
which we distinguish modular from random. In order to quantify the accuracy of this classification, we
consider the resulting false discovery rate (FDR) and false omission rate (FOR).\footnote{
	\[
	\text{FDR}=\frac{\text{FP}}{\text{TP}+\text{FP}}, \qquad\qquad
        \text{FOR}=\frac{\text{FN}}{\text{TN}+\text{FN}},
	\]
	where TP (true positive) is the number of correctly identified  base pairs, 
	FP (false positive) is the number of incorrectly
	predicted pairs that do not exist in the accepted
	structure, 
	TN (true negative) is the number of 
	pairs of bases that are  correctly identified as unpaired and 
	FN (false negative) is the number of base pairs in the accepted RNA structure 
	that are incorrectly predicted as unpaired.}
In our R\'{e}nyi-Ulam game variation, the expected values of FDR and FOR are the error rates $e_1$
and $e_0$ in case the truthful answer being yes and no, respectively.
Fig.~\ref{F:threshold} displays the error rates $e_0$ and $e_1$ as functions of $\theta$.
For $\theta=31$, we compute $e_0\approx 0.052$ and $e_1 \approx 0.007$, i.e. ~we have an error rate
of $ 0.052$ for rejecting and an error rate of $0.007$ for confirming a base pair.

\begin{figure}
	\centering
	\includegraphics[width=0.8\textwidth]{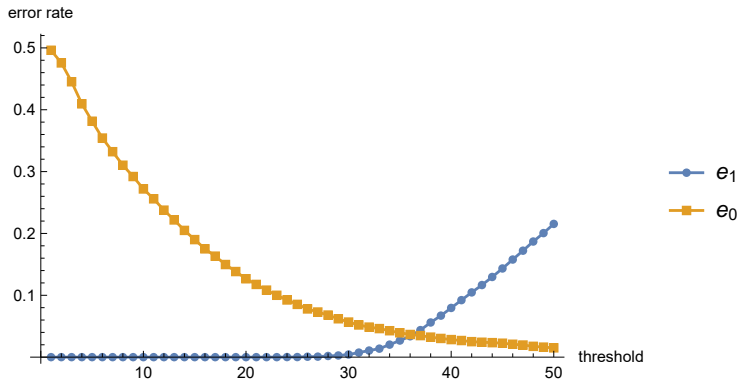}
	\caption
	{\small The error rates $e_0$ and $e_1$ as a function of the threshold $\theta$.
		We use the same sequences and structures as described in Fig.~\ref{F:shapeM}.
	}
	\label{F:threshold}
\end{figure}

\subsection{The oracle via experimental data }
\label{S:experiment}

The identification of base pairs is a fundamental and longstanding problem in RNA
biology~\citep{Hajdin:13,Weeks:15}. In ~\ref{S:oracle}, we summarize state-of-the-art
experimental approaches that provide  reliable solutions to the problem, and in particular
detail two methods, both of which utilize chemical probing~\citep{Mustoe:19,Cheng:17} and recover duplexes
with a false discovery rate less than $0.05$.

Successive queries recursively split a given ensemble of structures. This induced sequence of splits
can be embedded in a binary tree, and be viewed as a path from the root to a leaf. We shall discuss
this tree in detail in the next section.

\section{The ensemble tree}
\label{S:properties}

Given an input sample $\Omega$, we construct the ensemble tree $T(\Omega)$ having maximum level $L=11$,
recursively computing the maximum entropy base pairs as described in Algorithm~\ref{ensemble}.
In this section, we shall analyze the entropy of leaves in order to quantify the existence of a distinguished
structure and to identify the target. 

\subsection{Entropy}
\label{S:bolt2}

To quantify the uncertainty of an ensemble, we define 
the \emph{structural entropy} of an ensemble, $\Omega$, of an RNA sequence, $\mathbf{x}$,
as the Shannon entropy
\begin{equation*}
H(\Omega) = -\sum_{s\in \Omega} p(s) \log_2 p(s),
\end{equation*}
the units of $H$ being bits.
The sum is taken over all secondary structures $s$ of $\mathbf{x}$, and $p(s)$ denotes the
Boltzmann probability of the structure $s$ in the ensemble $\Omega$. The notion of structural
entropy is originated in thermodynamics and is usually regarded as a measure of disorder,
or randomness of an ensemble~\citep{Sukosd:13,Garcia:15}.

Given a sample $\Omega'$ of size $N$, the structural entropy has the upper bound $ \log_2
N$, that is, $H_{}(\Omega')$  reaches its maximum when all sampled structures are different.
Throughout the paper, we assume $N=1024$ and therefore $H_{}(\Omega')\leq 10$.

\begin{proposition}\label{C:bound2}
	Let $\Omega'$ be a sample having structural entropy $E$, where $0\leq E \leq 1$. 
	Then there exists one structure in $\Omega'$ having probability at least $f(E)$,
	where  $f(E)$  is the solution of the equation
	$$
	-p \log_2 p -(1-p)  \log_2 (1-p)=E
	$$
	satisfying $0.5\leq p \leq 1$. In particular, we have $f(1)=0.5$, $f(0.469) \approx 0.9$ and 
	$f(0.286) \approx 0.95$, see Fig.~\ref{F:fE}.
\end{proposition}
\begin{figure}
	\centering
	\includegraphics[width=0.6\textwidth]{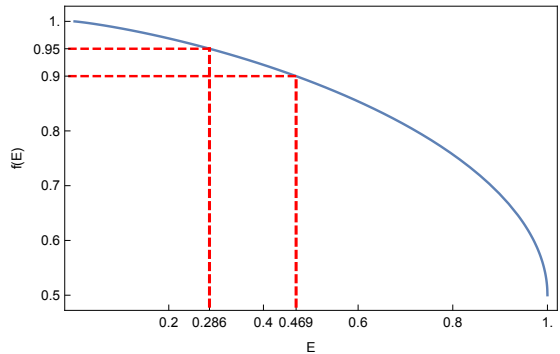}
	\caption
	{\small A sample with structural entropy $E$ contains a distinguished structure having probability
		at least $f(E)$.
	}
	\label{F:fE}
\end{figure}

Proposition~\ref{C:bound2} implies that a sample with small structural entropy contains a
distinguished structure and a proof is given in ~\ref{A:3}.
We refer to a sample having a distinguished structure of probability at least $\lambda$ as
being \emph{$\lambda$-distinguished}.
   
Next we quantify the reduction of a bit query on an ensemble.
Recall that the associated r.v. $X_{i,j}$  of a base pair $(i,j)$ 
partitions the sample $\Omega$ into two disjoint sub-samples $\Omega_0$ and
$\Omega_1$, where  $ \Omega_k=\{s\in  \Omega :X_{i,j}(s)=k\}$ ($k=0,1$).

The \emph{conditional entropy}, $H(\Omega|X_{i,j})$,  
represents 
the expected value of the entropies of the conditional distributions on $\Omega$,
averaged over the conditioning r.v. $X_{i,j}$ and can be
computed by 
\begin{equation*}
H(\Omega|X_{i,j})= (1-p_{i,j}) H(\Omega_0)+ p_{i,j} H(\Omega_1).
\end{equation*}
Then the \emph{entropy reduction} $R(\Omega,X_{i,j})$ of  $X_{i,j}$ on $\Omega$ 
is the difference between the \textit{a priori} Shannon entropy $H(\Omega)$ and the conditional
entropy $H(\Omega|X_{i,j}) $, i.e. 
\begin{equation*}
R(\Omega,X_{i,j})= H(\Omega)-H(\Omega|X_{i,j}). 
\end{equation*}
The entropy reduction quantifies the average change in information entropy from an ensemble
in which we cannot tell whether or not a certain structure contains $(i,j)$, to its bipartition
where one of its two blocks consists of structures that contain $(i,j)$ and the other being its
complement.
\begin{proposition}\label{P:1}
	The entropy reduction $R(\Omega,X_{i,j})$ of $X_{i,j}$ is given by the entropy
	$H(X_{i,j})$ of $X_{i,j}$, i.e. 
	\begin{equation}\label{Eq:IR}
	R(\Omega,X_{i,j})= H(X_{i,j}).
	\end{equation}
\end{proposition}

Proposition~\ref{P:1} queries a Bernoulli random variable inducing a split, reducing its average
conditional entropy exactly by the entropy of the random variable itself. In the context of the
R\'{e}nyi-Ulam game, Q asks a question that helps to maximally reduce the space of possibilities.
A proof of Proposition~\ref{P:1} is presented in ~\ref{A:4}.

The next observation shows that querying maximum entropy base pairs, induces a best possible balanced split
of the ensemble.

\begin{proposition}\label{P:subspace}
	Suppose that $X_{i,j}$ induces a partition of the ensemble $\Omega$ into sub-samples
	$\Omega_0^{i,j}$  and $\Omega_1^{i,j}$.
	Let $(i_0,j_0)$ be a maximum entropy base pair of $\Omega$.
	Then we have\\
	{\rm (1)} 
	$(i_0,j_0)$ minimizes the difference of the probabilities of the two sub-samples,
	\[
	|\mathbb{P}(\Omega_0^{i_0,j_0}) -\mathbb{P}(\Omega_1^{i_0,j_0})|\leq |\mathbb{P}(\Omega_0^{i,j}) -\mathbb{P}(\Omega_1^{i,j})|,
	\]
	for any $(i,j)$.\\
	{\rm (2)} $(i_0,j_0)$  maximizes  the entropy reduction $R(\Omega,X_{i,j})$  
	of  $X_{i,j}$ on $\Omega$,
	\[
	R(\Omega,X_{i_0,j_0}) \geq R(\Omega,X_{i,j}),
	\]	
	for any $(i,j)$.	 
\end{proposition}

Proposition~\ref{P:subspace}  first shows that the bit query about the maximum entropy base pair $X_{i_0,j_0}$ 
partitions the ensemble as balanced as possible, i.e.~into sub-samples having the minimum difference of their
probabilities. It furthermore establishes that the splits have minimum average structural entropy (or
uncertainty), since $X_{i_0,j_0}$ provides the maximum entropy reduction on the ensemble. Thus the query about
$(i_0,j_0)$ is the most informative among all bit queries.

Finally we  quantify the average entropy of sub-samples, $\Omega_{t}$, 
on the $t$-th  level  of the ensemble tree, and establish the existence of a distinguished
structure.
The analysis of entropies depends of course on the way the samples
are being constructed.        
To this end,
we construct the ensemble tree for two types of samples, 
one being unrestricted samples of random sequences, $\Omega$, and the other
utilizing $q$-Boltzmann sampling that incorporates
the signature of the target, $\Omega^q$, see Section~\ref{S:bolt}.

For unrestricted Boltzmann samples,
  the    structural entropy $H(\Omega_{t})$  of sub-samples on the $t$-th  level decreases, 
 as the  level  $t$ increases, see Fig.~\ref{F:SelStrpathLevel2}.
 In particular, the  average entropy $H(\Omega_{11})$ of leaf samples is 
 $0.328 $ and $0.147$, for sequences having $200$ and $300$ nucleotides,
 respectively.
 Proposition~\ref{C:bound2} guarantees that 
 the leaf $\Omega_{11}$  is  $0.90$-distinguished, i.e.  
 containing a distinguished structure with ratio at least  $0.90$,
 and $0.95$-distinguished for sequences of length $300$.

	  \begin{figure}
	  	\centering
	  	\includegraphics[width=0.8\textwidth]{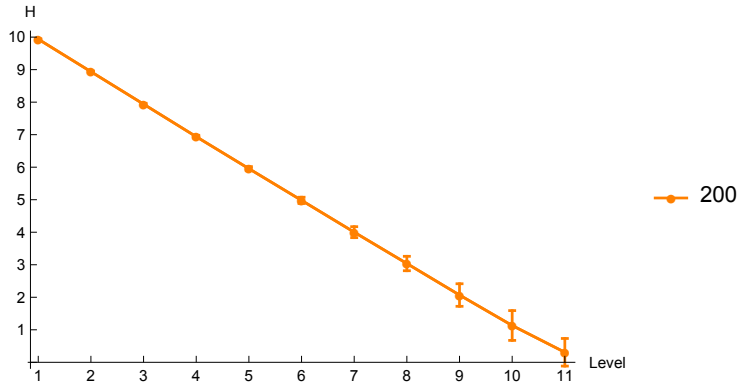}
	  	\caption{\small The average entropy of
	  		sub-samples  $H(\Omega_{t})$ on the $t$-th level.
	  		We randomly generate $10^3$ sequences of length  $200$, and sample $2^{10}$ structures
	  		together with a target structure $s$
	  		for each sequence. 
	  	}
	  	\label{F:SelStrpathLevel2}
	  \end{figure}

For $q$-Boltzmann samples  $\Omega^q$ of structures having signature distance to the target $s$ at most $qn$,
the small entropy of the leaf and the high ratio of the distinguished structure are
robust over a range of $q$-values, see Fig.~\ref{F:EntLeafq}.
We also observe that,
for longer sequences,
the entropy is smaller,
and therefore the ratio of the distinguished structure is higher.

    \begin{figure}
    	\centering
    	\includegraphics[width=0.8\textwidth]{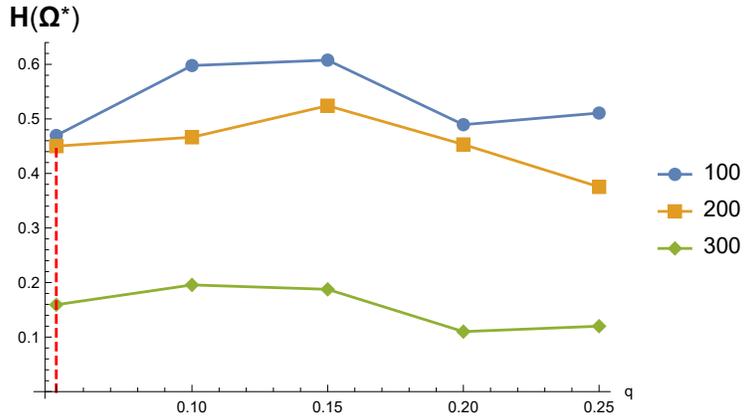}
    	\caption
    	{\small  The structural entropy $H(\Omega^q_{11})$  of the leaf sub-samples
    		for different  $q$-values.
    		We randomly generate $10^3$ sequences of length  $100$, $200$ and $300$. 
    		For each sequence, we then generate a $q$-Boltzmann sample $\Omega^q$
    		of $2^{10}$ structures
    		together with a target $s$.
    		The red dashed line  denotes $q$-samples having $q=0.05$,
    		which is tantamount to
    		Boltzmann samples $\Omega_{\text{probe}}$ incorporating the probing data via
    		pseudo-energies.
    	}
    	\label{F:EntLeafq}
    \end{figure}


\subsection{Target Identification}
\label{S:StaProperty}

Any leaf of the ensemble tree exhibiting a structural entropy less than one, contains, by Proposition~\ref{C:bound2},
a distinguished structure. Successive queries produce a unique, distinguished leaf, $\Omega^*$ which, with
high probability, contains structures that are compatible with the queries. Let $s^*$ be the distinguished
structure in $\Omega^*$, and $s$ denote the target.

In this section, we shall analyze this probability, $\mathbb{P}(s\in \Omega^*)$, as well as $\mathbb{P}(s^*=s)$ and
$\mathbb{P}(s^*=s\mid s\in \Omega^*)$, see Table~\ref{Tab:s1}. For the path identification to the leaf $\Omega^*$,
we consider the error rates $e_0=0.05$ and $e_1=0.01$ computed in Section~\ref{S:modular}.

As detailed in Section~\ref{S:modular}, these probabilities depend on the error rates $e_0$ and $e_1$, and
since these errors occur independently, we derive
$\mathbb{P}(s\in \Omega^* )=(1-e_0)^{l_0} (1-e_1)^{l_1}$,
where $l_0$ and  $l_1$ denote the number of No-/Yes-answers to queried base pairs along the path,
respectively.  
Fig.~\ref{F:YesDistribution} displays the distribution of $l_1$. We observe that $l_1$ has
a mean around $5$,
i.e., the probabilities of queried base pairs being confirmed and being rejected
are roughly equal.
For $l_0=l_1=5$,  we have a theoretical estimate  $\mathbb{P}(s\in \Omega^* )\approx 0.736$.
In Fig.~\ref{F:estimateLeaf}  we present  that $\mathbb{P}(s\in \Omega^* )$ decreases 
as  the error rate $e_0$ increases, for fixed $e_1=0.01$.

\begin{figure}
	\centering
	\includegraphics[width=0.8\textwidth]{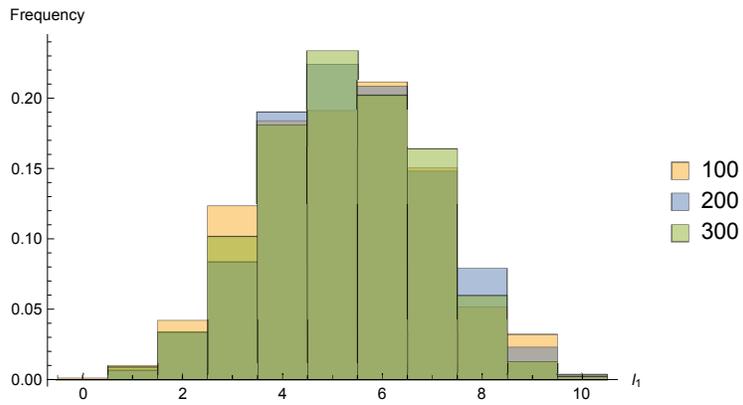}
	\caption
	{\small  The distributions of $l_1$, the number of queried base pairs on the path that are confirmed by the target structure.
		We generate unrestricted Boltzmann samples for random sequences of different lengths.
	}
	\label{F:YesDistribution}
\end{figure}

\begin{figure}
	\centering
	\includegraphics[width=0.8\textwidth]{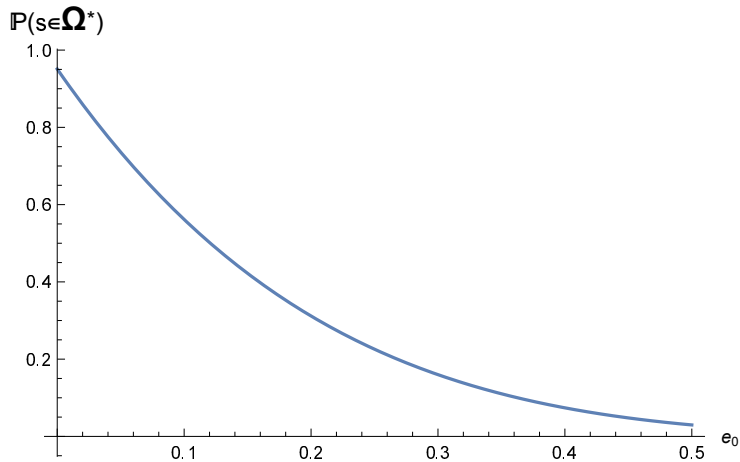}
	\caption
	{\small  The probability $\mathbb{P}(s\in \Omega^* )$ as a function of 
	the error rate $e_0$, for fixed $e_1=0.01$ and $l_0=l_1=5$.
	}
	\label{F:estimateLeaf}
\end{figure}

For (unrestricted) Boltzmann samples generated from random sequences, we present the probability
$\mathbb{P}(s\in \Omega^*)$ of the leaf containing the target is greater than $74\%$, which agrees with the above 
theoretical estimate. Note that this amounts
to having no probing data as a constraint for the sampled structures, a worst case scenario, so to speak.

\begin{table}[htbp]
	\caption{ Key observables.} \label{Tab:s1}
	\begin{tabular}{cc}
		\hline\noalign{\smallskip}
		\small 	Quantity  &\small Description \\
		\noalign{\smallskip}\hline\noalign{\smallskip}
		\small 	$\mathbb{P}(s\in \Omega^*)$  &\small the probability  of the target being in the leaf\\
		\small 	$\mathbb{P}(s^*=s)$  &\small  the probability  of  the distinguished structure being  identical to the target \\
		\small  $\mathbb{P}(s^*=s\mid s\in \Omega^*)$	&\small  the probability of correctly identifying the target, given that it is in the leaf
		\\   	    
		\noalign{\smallskip}\hline\noalign{\smallskip}	            
	\end{tabular}
\end{table}

Furthermore, the probability that the distinguished structure is identical to the target is approximately
unchanged, see Table~\ref{Tab:sum}. $\mathbb{P}(s^*=s\mid s\in \Omega^*)$ indicates, that once we are in the
correct leaf, the chance of correctly identifying the target increases to  $94\%$ for sequences of length $300$.
Accordingly, the key factor is the correct identification of the leaf $\Omega^*$.

\begin{table}[htbp]
	\caption{Target identification:
	we randomly generate $10^3$ sequences of length $n$ and Boltzmann sample $2^{10}$ structures
	together with a target structure $s$ for each sequence. We display mean and standard deviation.
         } \label{Tab:sum}
	\begin{tabular}{cccc}
		\hline\noalign{\smallskip}
		\small 	&\small $n=100$   &\small $n=200$  & \small $n=300$ \\
		\noalign{\smallskip}\hline\noalign{\smallskip}
				\small  $\mathbb{P}(s\in \Omega^*)$	&\small $0.768 \pm 0.178$             
		&\small $0.742 \pm 0.192$			 
		&\small $0.751 \pm 0.187$         	\\     
			\small  $\mathbb{P}(s^*=s)$	&\small $0.669 \pm 	0.222$             
		&\small $0.646\pm 0.229$			 
		&\small $0.706 \pm 0.208$         	\\   	 
		\small $\mathbb{P}(s^*=s\mid s\in \Omega^*)$ &\small $0.871 \pm 0.288$             
		&\small $0.871 \pm 0.309$ 			 
		&\small $0.940 \pm 0.277$         		         	\\    
		\noalign{\smallskip}\hline\noalign{\smallskip}	            
	\end{tabular}
\end{table}

For $q$-Boltzmann samples $\Omega^q$ filtered by signature distance $\le q n$ we observe the following:
the probability $\mathbb{P}(s\in \Omega^{*} )$ of the leaf to contain the target is greater than $70\%$ 
is robust over a range of $q$-values, see Fig.~\ref{F:ProbLeafq}.
As expected, as $q$ increases, the probability of the target being in the correct leaf decreases, due to the
fact that the $q$-samples become less constraint by the probing data.

In particular, 
we observe that, for $q=0.05$ and sequences of length $300$,  the probability of the ensemble tree 
correctly identifying the target in the  leaf is greater than $90\%$, see Fig.~\ref{F:ProbLeafq} (red dashed line).
As  the Boltzmann ensembles incorporation of probing data via
pseudo-energies result in a  $q$-value of $0.05$, this translates into
$\mathbb{P}(s\in \Omega^{*})\ge 90\%$ for such ensembles generated by such restricted Boltzmann samplers
for sequences of length $300$.

We demonstrate that the ensemble tree  localizing the target 
with high fidelity is robust, across  samples of sequences having 
various lengths and different  signature  filtration $q$.
Fig.~\ref{F:ProbEq} (LHS) shows that the ensemble tree for longer sequences has a  higher 
chance of identifying the target.
Once we are in the correct leaf, the chance of  correctly distinguishing the target
significantly increases, from around $75\%$ to over $94\%$ 
in the case of sequences having $200$ nucleotides,
see Fig.~\ref{F:ProbEq} (RHS).

\begin{figure*}
	\centering
	\includegraphics[width=0.8\textwidth]{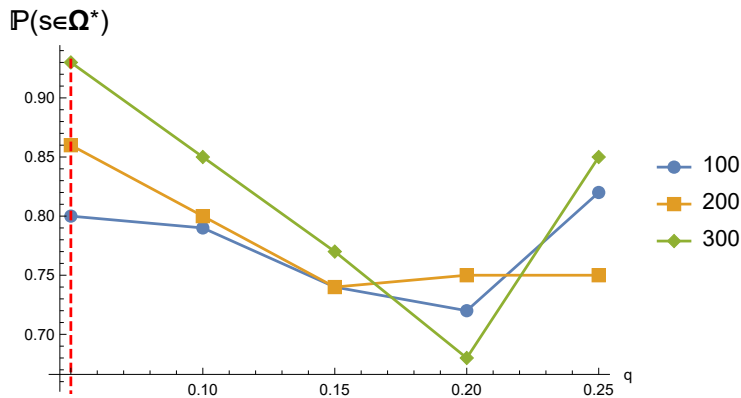}
	\caption{\small The probability $\mathbb{P}(s\in \Omega^{*} )$
		of being in  the correct leaf for different $q$-values.
	We use the same sequences and $q$-Boltzmann samples as described in  Fig.~\ref{F:EntLeafq}.	
  		The red dashed line  denotes $q$-samples having $q=0.05$,
which is tantamount to
Boltzmann samples $\Omega_{\text{probe}}$ incorporating the probing data via
pseudo-energies.}
	\label{F:ProbLeafq}
\end{figure*}

\begin{figure*}
	\begin{tabular}{cc}
		\includegraphics[width=0.45\textwidth]{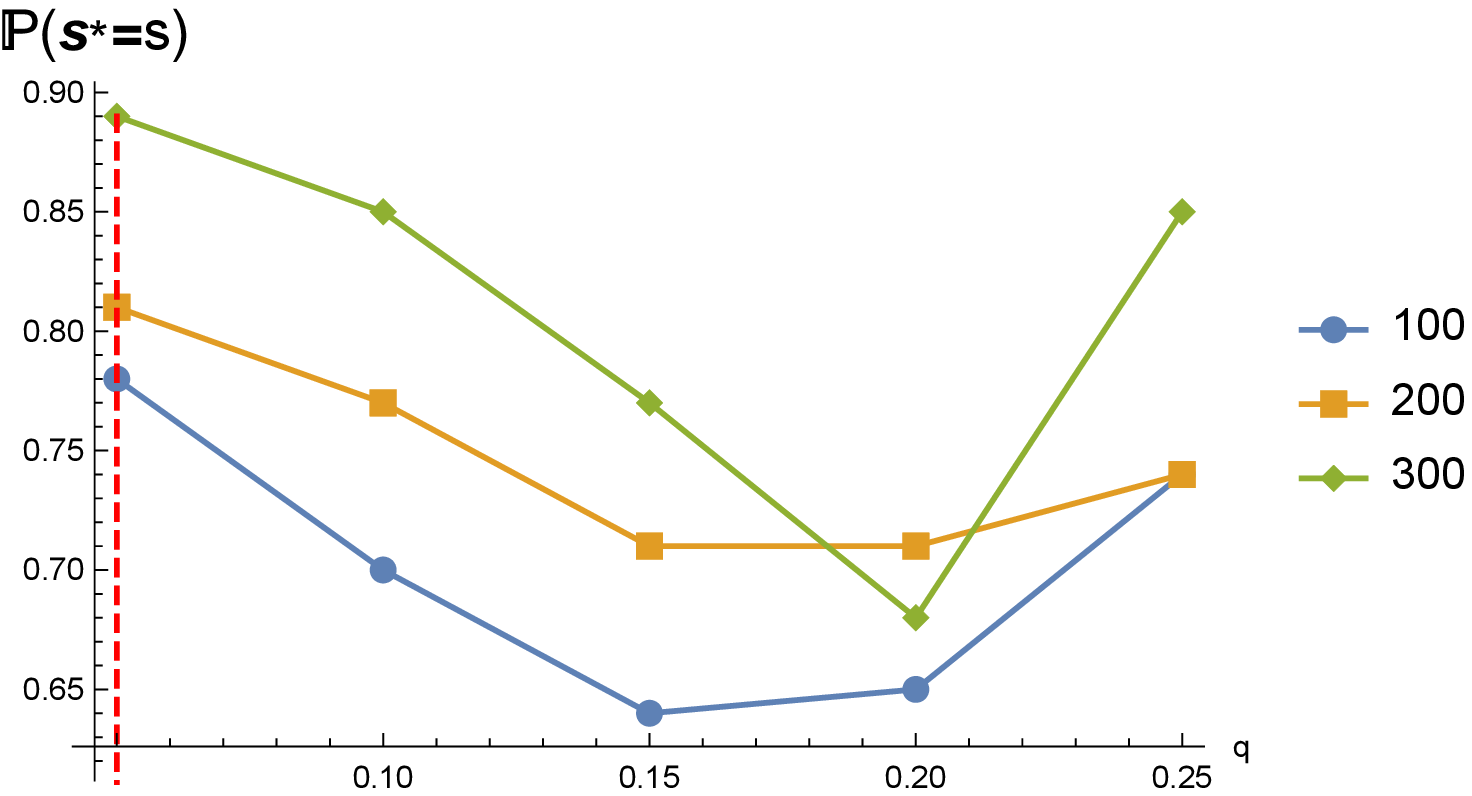}&
		\includegraphics[width=0.45\textwidth]{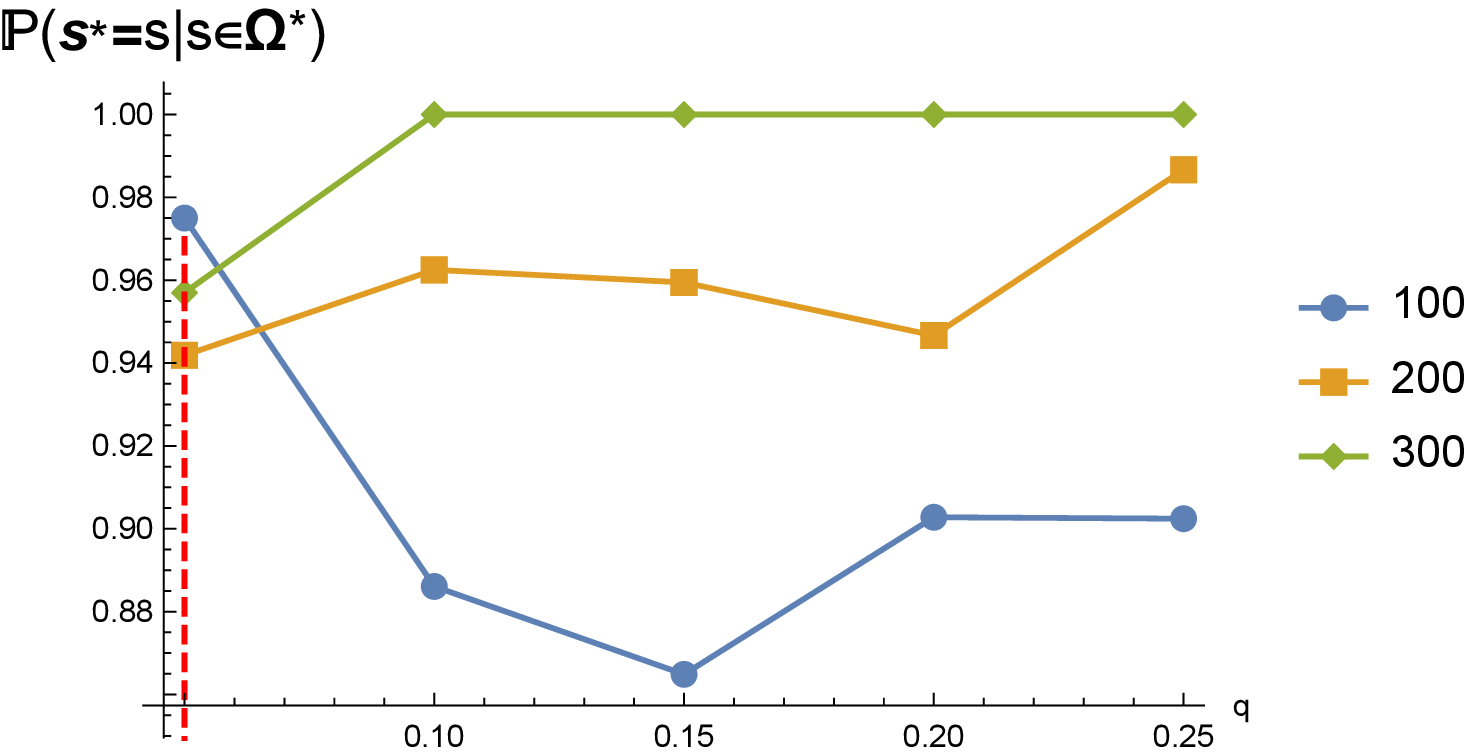}
		\\
	\end{tabular}
	\caption{\small 
		The probabilities $\mathbb{P}(s^*=s)$	 (LHS)
		and $\mathbb{P}(s^*=s\mid s\in \Omega^*)$ (RHS)
		of correctly identifying the target, either  in general or
		conditioning on being in the correct leaf.
		We use the same sequences and $q$-Boltzmann samples as described in  Fig.~\ref{F:EntLeafq}.	
  		The red dashed line  denotes $q$-samples having $q=0.05$,
which is tantamount to
Boltzmann samples $\Omega_{\text{probe}}$ incorporating the probing data via
pseudo-energies.	}
	\label{F:ProbEq}
\end{figure*}

As mentioned above, the key is the correct identification of the leaf containing the
target, and its distinguished structure to coincide with the latter. These events are quantified via
$\mathbb{P}(s\in \Omega^{*} )$ and $\mathbb{P}(s^*=s)$, which depend on the error rates $e_0$ and $e_1$.

These error rates 
can be reduced by asking the same query  repeatedly.  
In our R\'{e}nyi-Ulam  game, repeating the same query 
is tantamount to performing the same experiment multiple times.
It is reasonable to assume that   experiments are performed independently
and thus errors occur randomly.
Intuitively, repeated experiments reduce 
errors originated from the noisy nature of experimental data.
Utilizing Bayesian analysis, we show that,  if we get the same  answer to the query twice,
the error rates would become significantly smaller, for example,
$e_0^{[2]} =0.003$ and $e_1^{[2]} =0.00005$, see ~\ref{A:error}.

In principle, we can reduce the error rates by repeating the same query $k$ times.
The error rates  would approach to $0$ as $k$ grows to infinity. 
 In this case, $\mathbb{P}(s\in \Omega^{*} ) \approx 1$, i.e. 
 the leaf always contains the target.
The fidelity of 
the distinguished structure $\mathbb{P}(s^{*}=s ) $ increases from $70\%$ to $94\%$
for sequences of length $300$.


\section{Discussion}
\label{S:dis}

In this paper we propose to enhance the method of identifying the target structure based on RNA probing data.
To facilitate this we introduce the framework of ensemble trees in which a sample derived from the partition
function of structures is recursively split via queries using information theory.
Each query is answered based on either RNA folding data in combination with chemical probing, employing
modularity of RNA structures, see Section~\ref{S:modular} or, alternatively, directly using experimental methods~\citep{Mustoe:19,Cheng:17}.
The former type of inference can be viewed as a kind of localization of probing data, relating local
to global data by means of structural modularity. We show that within this framework it is possible to identify
the target with high fidelity and that this identification requires a small number of base pairs to be queried.
In particular we present that, for the Boltzmann ensembles incorporating
	probing data via pseudo-energies, the probability of the ensemble tree 
 identifying the  correct leaf that contains the target is greater than $90\%$, see Section~\ref{S:StaProperty}.

In our framework, the key factor is the correct identification of the leaf that contains the target.
Fig.~\ref{F:SelStrpatherror} displays the average base-pair distances $d_{\text{bp}}(s,\Omega_{t})$
\footnote{Here 
	$
	d_{\text{bp}}(s,\Omega)= \sum_{s'\in \Omega} p(s') d_{\text{bp}}(s,s').
	$}
between the target structure $s$ and the $t$-th sub-sample $\Omega_{t}$ on the path. We contrast three scenarios, first
the expectation being taken over all ensemble trees (blue), the set of ensemble trees in which the
leaf containing the target is identified (green) and its complement (orange). 
We here present that the correct identification of the leaf containing the
target significantly reduces the distance between  the target and the sub-samples.

\begin{figure}
	\centering
	\includegraphics[width=0.8\textwidth]{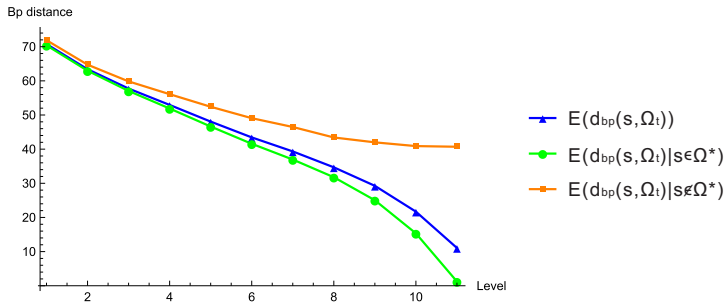}
	\caption
	{\small The average base-pair distance $d_{\text{bp}}(s,\Omega_{t})$ between  the target   $s$ and the sub-sample $\Omega_{t}$ on the path.
		The expectation is taken over all ensemble trees (blue), the set of ensemble trees in which the
		leaf containing the target is identified (green) and its complement (orange).
	The computation is based on the   Boltzmann  samples of sequences of length $300$.
	}
	\label{F:SelStrpatherror}
\end{figure}

Our framework is based on two assumptions.
The first is sampling from the Boltzmann ensemble of structures.
This assumption is important, as for an arbitrary sample, 
the leaf of the ensemble tree does not always contain
 a distinguished structure.
By quantifying the distinguished structure via the flow of entropies of sub-samples on the path, 
we contrast three classes of samples,
the first being  a Boltzmann sample (B-sample), the second  a uniform sample  (U-sample) and the third an E-sample\footnote{consisting of 
	$N$ different structures with the uniform distribution, 
	each structure  containing only one base pair.},
see Fig.~\ref{F:SelStrpathLevel}.
We present that,
in a Boltzmann sample,
the entropies of sub-samples on the $t$-th level
decrease  much more sharply than those in  the latter two classes, see Fig.~\ref{F:SelStrpathLevel} (LHS).
 In particular, the latter two  produce
 leaves exhibiting an average entropy greater than $1$,
 i.e. not containing a distinguished structure.
 As proved in Proposition~\ref{P:1},
 the entropy reduction equals to the entropy of the queried base pair.
Fig.~\ref{F:SelStrpathLevel} (RHS) explains the reason for the significant reduction,
that is, 
the maximum entropy base pairs  in  Boltzmann samples have
entropy close to $1$ on each level, implying that the bit queries  split
the ensemble roughly in half each time.
The latter two types of samples do not exhibit this phenomenon.

\begin{figure*}
	\begin{tabular}{cc}
		\includegraphics[width=0.45\textwidth]{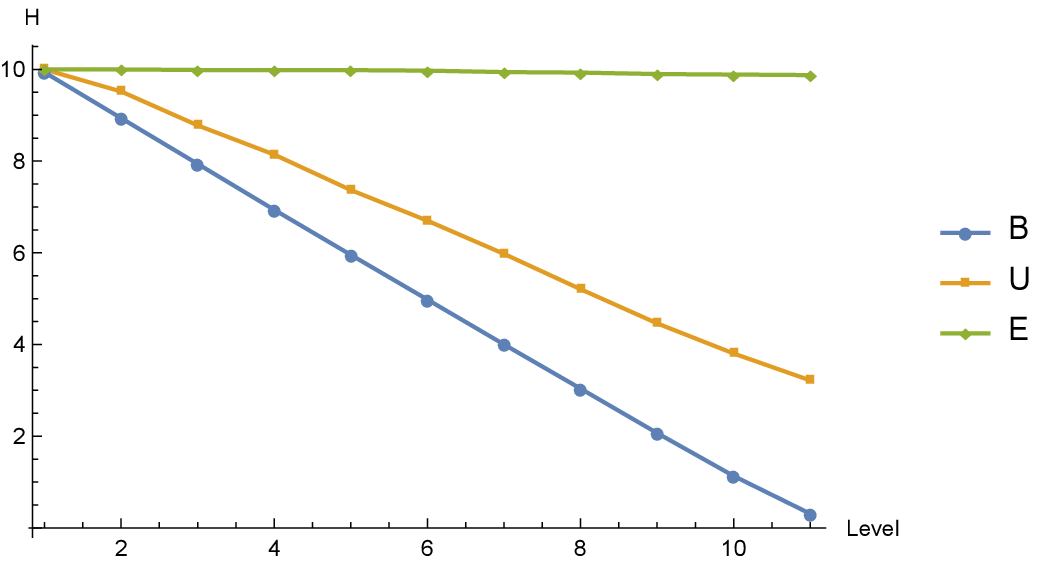}&
		\includegraphics[width=0.45\textwidth]{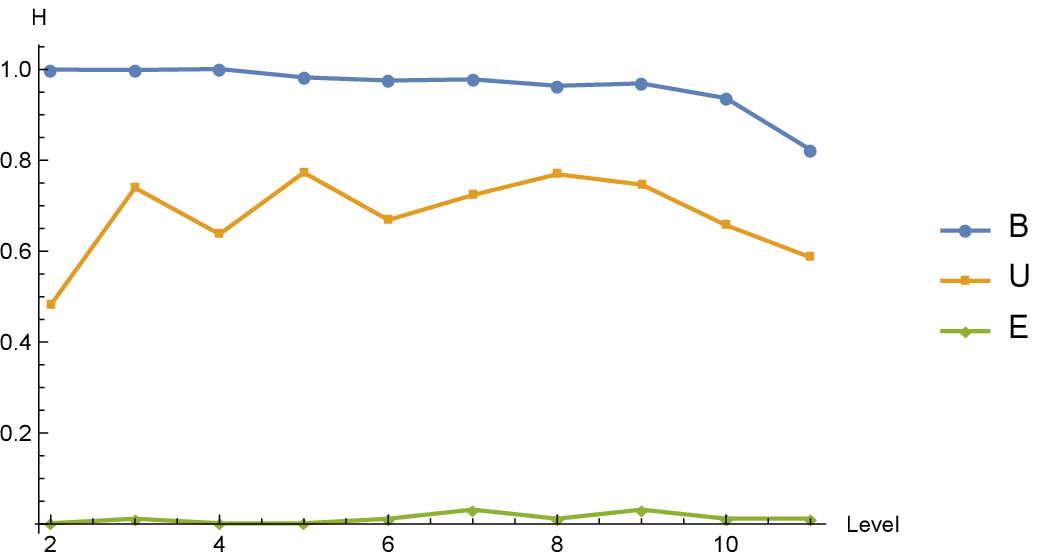}
		\\
	\end{tabular}
	\caption{\small The average entropy of
		sub-samples  $H(\Omega_{t})$ (LHS) and  queried base pairs  $H(X_{t})$ (RHS)   on the $t$-th level of the
		ensemble tree.
We contrast the ensemble trees obtained from
 a Boltzmann sample (B, blue), a uniform sample (U, orange), or an E-sample (E,  green),
 which is comprised  of $2^{10}$ distinct  structures, each containing only one base pair.
 For the former two	types of samples,
 we randomly generate $10^3$ sequences of length  $200$.
	For each sequence, we 
	 sample $2^{10}$  structures 
	together with a target structure $s$, according to the Boltzmann 
	or uniform distributions.
	}
	\label{F:SelStrpathLevel}
\end{figure*}

The second assumption is that the target is contained in the sample.
This assumption can be validated by generating  samples of larger size,
and checking whether or not the distinguished structure is reproducible.

Accordingly, the  probability and entropy  of a base pair  is calculated 
in the context of the entire ensemble,
 and thus  the ensemble tree together with  maximum entropy base pairs.
\cite{Garcia:15} show that
 the structural entropy of the entire Boltzmann ensemble  
 is asymptotically  linear in $n$, i.e.  
 $H(\Omega_{\text{entire}})\approx 0.07 n$.
 Since each queried base pair reduces the entropy by approximately $1$
 and the reduction is additive by construction,
the ensemble tree would require approximately $0.07 n$ queries to 
 identify a leaf that has entropy smaller than $1$ and contains a distinguished structure.

\begin{figure}
	\centering
	\includegraphics[width=0.6\textwidth]{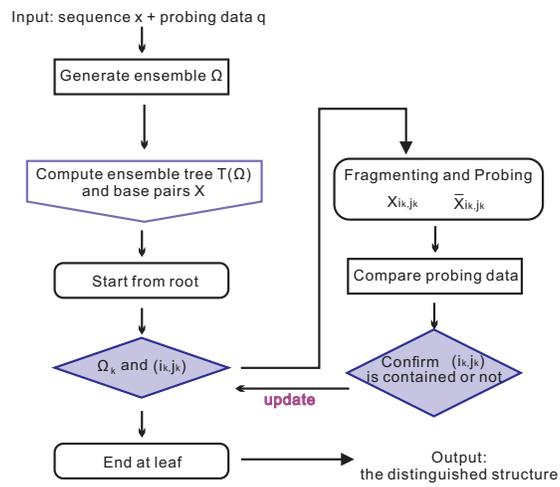}
	\caption
	{\small 
		The workflow diagram of our fragmentation process.
	}
	\label{F:workflow}
\end{figure}

\begin{figure}
	\centering
	\includegraphics[width=0.6\textwidth]{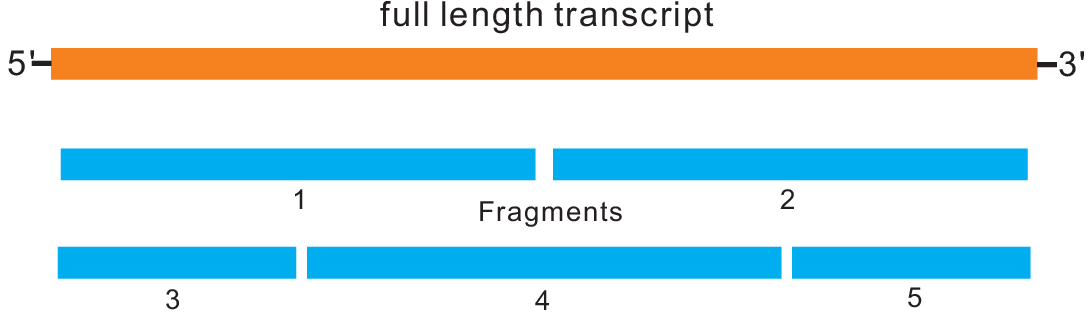}
	\caption
	{\small The  fragmentation by~\cite{Novikova:13}.
	}
	\label{F:fragment}
\end{figure}

Equipped with the ensemble tree and chemical probing,  our framework provides  a  fragmentation process 
combining ''local'' probing profiles with the ''global'' one via modularity. 
For each queried base pair, our fragmentation  subsequently   splits the sequence,
and determines  the presence of base pairs
via comparing probing profiles (Section~\ref{S:modular}).
Fig.~\ref{F:workflow} demonstrates the workflow of the fragmentation process,
see ~\ref{A:5}.
\cite{Novikova:13} developed 
a different fragmentation method
for  determining
the secondary structure of lncRNAs.
Their  approach
applies chemical probing of the entire RNA, followed by probing of
certain overlapping fragments, see Fig.~\ref{F:fragment}.
Regions of each fragment exhibiting similar probing profiles
are folded independently,
and  combined in order to obtain the entire structure.
At a fundamental level, our fragmentation is different from their approach
in that we allow bases from two non-contiguous fragments to pair.
Their  approach  prohibits long-range pairs, such as connecting fragments $3$ and $5$ in Fig.~\ref{F:fragment}.
As a consequence, our method is well suited to deal with
the long-range base pairings,
 whose existence  
has been shown experimentally~\citep{Lai:18} as well as
theoretically~\citep{Li:18,Li:19}.

For a sample of RNA  pseudoknotted structures,
the ensemble tree in our framework can still be computed.
However, the structure modularity no longer holds in the pseudoknot case.
The reason is that a pseudoknot loop could intersect in more than one base pair with other loops,
see Fig.~\ref{F:loop} (RHS).
The fragmentation  with respect to a base pair involved in a pseudoknot could
affect several loops, each contributing to the free energy. 
The  change of loop-based  energy  could lead to  splits  folding into a different configuration
compared to  the  full transcript.
Nevertheless, it would be interesting to find out other experimental methods to 
facilitate our framework for RNA  pseudoknotted structures.

\begin{center}
	{\bf ACKNOWLEDGMENTS} We want to thank Christopher Barrett for stimulating discussions and the staff of the
	Biocomplexity Institute \& Initiative  at University of Virginia for their great support.
We would like to thank Dr. Kevin Weeks for pointing out their recent work~\citep{Mustoe:19}.
 Many thanks to Qijun He, Fenix Huang, Andrei Bura,  Ricky Chen, and  Reza Rezazadegan for discussions.

\end{center}


\begin{center}
	{\bf AUTHOR DISCLOSURE STATEMENT}
\end{center}

The authors declare that no competing financial interests exist.





\appendix
	
	
	\section{RNA secondary structures}
	\label{A:1}
	
Most computational approaches of RNA structure prediction reduce to a class of
	coarse grained structures, i.e.  
	the RNA secondary structures~\citep{Waterman:78s,Waterman:79a,Waterman:78aa,Howell:80,Waterman:93}.
	These are contact structures via
	abstracting from the actual spatial
	arrangement of nucleotides.
	An RNA secondary structure can be represented as a \emph{diagram},   
	a labeled graph over the vertex set $\{1, \dots, n\}$ 
	 whose vertices are arranged in a horizontal line and arcs are drawn in the upper half-plane. 
	Clearly, vertices correspond to nucleotides in the primary sequence 
	and arcs correspond to  the  Watson-Crick \textbf{A-U}, \textbf{C-G} and wobble \textbf{U-G}
	base pairs. 
Two arcs 	$(i_1,j_1)$ and $(i_2,j_2)$ form a pseudoknot if they  cross, i.e. 
 the nucleotides appear in the order $i_1<i_2<j_1<j_2$ in the primary sequence.
	   An
RNA \emph{secondary structure} is  a diagram 
without pseudoknots.

We define two distances for comparing two structures,
the base-pair and signature distances.

The base-pair distance utilizes
a representation of a secondary structure $s$  as a bit string $\mathbf{b}(s)=b_1b_2\ldots b_{l}$,
where $l$ denotes the number of  all possible base pairs,
and $b_k$ is a bit.
Given the arc set $E$ equipped with the lexicographic order,
we define $b_k=1$ if $s$ contains the $k$-th base pair in $E$, otherwise $b_k=0$.
The \emph{base-pair distance} $d_{\text{bp}}(s,s')$  
between two structures $s$ and $s'$  is 
the Hamming distance between their corresponding bit strings $\mathbf{b}(s)$ and $\mathbf{b}(s')$ .


The \emph{$0$-$1$ signature} (or simply signature) of a structure $s$,
is a vector 
$\mathbf{q}(s)=(q_1,q_2,\ldots,q_n)$, 
where $q_k=1$ when the $k$-th base is unpaired in $s$, otherwise $q_k=0$.
The \emph{signature distance} $d_{\text{sn}}(s,s')$  between two structures $s$ and $s'$ is defined as the
Hamming distance between their corresponding $0$-$1$ signatures $\mathbf{q}(s)$ and $\mathbf{q}(s')$.
By construction,   the $0$-$1$ signature  of a secondary structure   mimics
its probing signals,
and the signature distance measures the similarity between the probing profiles of two structures.
By observing that each bit corresponds to two base-pairing end, 
we derive $d_{\text{sn}}(s,s') \leq 2 d_{\text{bp}}(s,s')$ for any  $s$ and $s'$.


\section{Energy model}
\label{A:energy}

Computational prediction of RNA secondary structures is mainly driven
by loop-based energy models~\citep{Mathews:99,Mathews:04}. 
The key
assumption of these approaches 
is that the free energy $E(s)$ of an RNA secondary structure
$s$, is estimated by the sum of energy contributions $E(L) $ from its
individual loops $L$, $E(s)=\sum_{L} E(L)$.

According to thermodynamics,
the free energy reflects not only the overall stability of the structure, 
but also its probability appearing in thermodynamic
equilibrium. This leads to the Boltzmann sampling~\citep{Ding:03,Lorenz:11} of 
secondary structure  based on their equilibrium 
probabilities, whose computation can be facilitated by
the partition function~\citep{McCaskill}.

In this model, the energy contribution of a base pair 
depends on the two adjacent loops that intersect at the base pair, see Fig.~\ref{F:loop} (LHS).
Note that, in a pseudoknot, since two adjacent loops may intersect at several base pairs,
and thus the energy contribution of a base pair could affect several loops, see Fig.~\ref{F:loop} (RHS).

\begin{figure*}
	\begin{tabular}{cc}
		\includegraphics[width=0.45\textwidth]{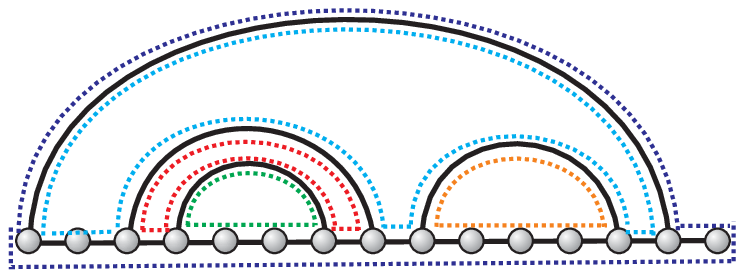}&
		\includegraphics[width=0.45\textwidth]{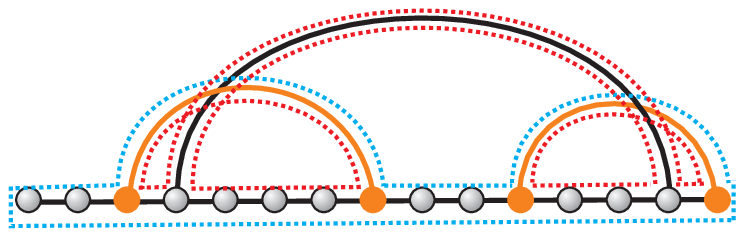}
		\\
	\end{tabular}
	\caption
	{\small  The loop-based decomposition of a secondary structure (LHS) and a pseudoknot (RHS).
		LHS: two adjacent loops intersect at one base pair. RHS: two pseudoknot loops meet at two base pairs (orange).
	}
	\label{F:loop}
\end{figure*}


\section{Chemical probing}
\label{A:probe}

The basic idea of RNA structure probing
is that  chemical probes react differently with paired or unpaired nucleotides.
More reactive regions of the RNA are likely to be single
stranded and less reactive regions are likely to be base paired. 
Thus every nucleotide in a folded RNA sequence can be assigned a reactivity score,
which depends on the type of chemical or enzymatic  footprinting experiments and 
the strength of the reactivity. 
It is rarely of absolute certainty, whether or not a specific position is unpaired, or paired;
instead, the method produces a probability. The probing data thus produce a vector of probabilities. 
Several competing methods have been developed to convert the footprinting data for each nucleotide
into a probability.
Probing data has been further incorporated into  RNA folding algorithms by adding a pseudo-energy
term, $\Delta G(s) $, to the free energy~\citep{Deigan:09, Washietl:12,Zarringhalam:12}, i.e. 
\[
E_{\text{probe}}(s)= E(s)+ \Delta G(s).
\]
This term engages in the folding process as follows: 
while positions where structure prediction and experiment data agree with
each other are rewarded by a negative pseudo-energy,
mismatching  locations receive a  penalty by way of a positive term.
This is tantamount to shifting the partition function in such a way that
 the equilibrium distribution of
structures in $\Omega_{\text{probe}}$ 
favors those that agree with the data.


\section{$q$-Boltzmann sampler
}
\label{S:probedata}

Here we incorporate the signature of a target via restricted Boltzmann sampling structures with the signature distance filtration.

We first analyze the signature distances in two classes of Boltzmann samples,
one being unrestricted, $\Omega$, and the other being restricted  
 $\Omega_{\text{probe}}$ that incorporates the signature of the target via pseudo-energies.

For both types of samples,
the distribution of the signature  distance 
between the target $s$ and the  ensemble is approximately normal, Fig.~\ref{F:SnDistESB}. 
The means and variances  of the normalized signature  distance are shown in Table~\ref{Tab:SnDistESB}.
It shows that, while the average signature  distance between the target and  the  unrestricted sampled structure 
is around $0.21 n$,
integrating the  signature of the target reduces the distance to $0.03 n$.
This indicates that the incorporation of the signature improves the accuracy of the Boltzmann sampler 
identifying  the target.

\begin{figure*}
	\begin{tabular}{cc}
	\includegraphics[width=0.45\textwidth]{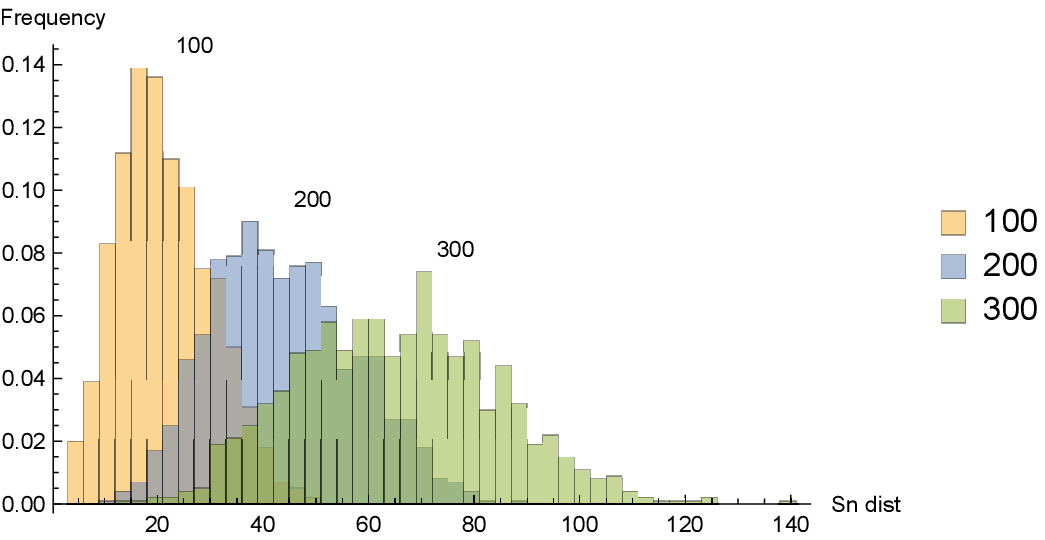}&
\includegraphics[width=0.45\textwidth]{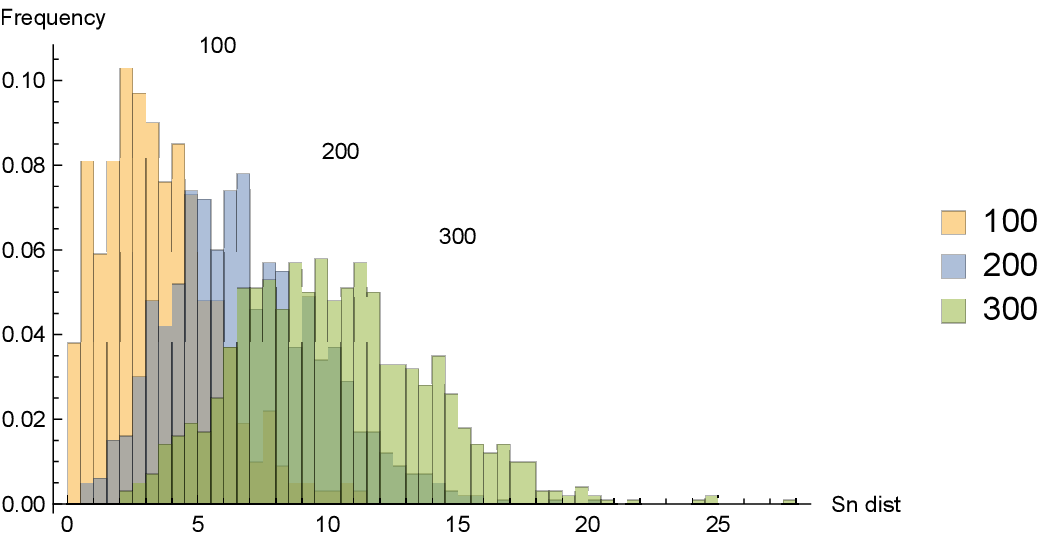}
	\\
\end{tabular}
	\caption
	{\small  The distributions of the signature  distances $d_{\text{sn}}(s,\Omega)$ and $d_{\text{sn}}(s,\Omega_{\text{probe}})$
		between the target $s$ and two types of Boltzmann ensembles. LHS utilizes unrestricted samples $\Omega$  of structures, and RHS uses
		samples  $\Omega_{\text{probe}}$  incorporating the signature of the target  structure $s$.
	}
	\label{F:SnDistESB}
\end{figure*}

\begin{table}[htbp]
	\caption{The means and variances of the  normalized signature  distances 
		between the target $s$ and the Boltzmann samples $\Omega$, $\Omega_{\text{probe}}$ or $\Omega^q$.
For $\Omega$ and $\Omega_{\text{probe}}$,	 we utilize the same Boltzmann samples as described in Fig.~\ref{F:SnDistESB}.
		Values following the $\pm$ symbols are the standard deviation of the sampling errors. } \label{Tab:SnDistESB}
	\begin{tabular}{cccc}
		\hline\noalign{\smallskip}
		\small 	&\small $n=100$   &\small $n=200$  & \small $n=300$ \\
		\noalign{\smallskip}\hline\noalign{\smallskip}
		\small $d_{\text{sn}}(s,\Omega)/n$	&\small $0.214 \pm 0.088$             
		&\small $0.219\pm 0.068$ 			 
		&\small $0.217 \pm 0.063$         		         	\\    
		\small $d_{\text{sn}}(s,\Omega_{\text{probe}})/n$	&\small $0.035 \pm 0.021$             
		&\small  $0.034\pm 0.015$ 			 
		&\small $0.034\pm 0.012$        		         	\\    
			\small $d_{\text{sn}}(s,\Omega^{0.05})/n$	&\small $0.031 \pm 0.008$             
		&\small  $0.038\pm 0.012$ 			 
		&\small $0.037\pm 0.018$        		         	\\    
					\small $d_{\text{sn}}(s,\Omega^{0.1})/n$	&\small $0.074 \pm 0.014$             
		&\small  $0.080\pm 0.015$ 			 
		&\small $0.087 \pm 0.010$        		         	\\    
					\small $d_{\text{sn}}(s,\Omega^{0.15})/n$	&\small $0.098 \pm 0.021$             
		&\small  $0.116 \pm 0.018$ 			 
		&\small $0.123 \pm 0.011$        		         	\\   
						\small $d_{\text{sn}}(s,\Omega^{0.20})/n$	&\small $0.127 \pm 0.034$             
		&\small  $0.144 \pm 0.027$ 			 
		&\small $0.157 \pm 0.020$        		         	\\     
							\small $d_{\text{sn}}(s,\Omega^{0.25})/n$	&\small $0.144 \pm 0.043$             
		&\small  $0.167 \pm 0.038$ 			 
		&\small $0.180 \pm 0.029$        		         	\\     
		\noalign{\smallskip}\hline\noalign{\smallskip}	            
	\end{tabular}
\end{table}

The above analysis motivates us to
 introduce a $q$-Boltzmann sampler for structures with signature distance  filtration.
For any fraction $q\in (0,1)$, let $\Omega^q$ denote the  restricted Boltzmann ensemble of 
structures having signature distance to the target at most $q\cdot n$,
i.e., $\Omega^q=\{ s'| d_{\text{sn}}(s',s) \leq q\cdot n \}$.
The enhanced Boltzmann sampling can be implemented 
by partition function~\citep{McCaskill} and stochastic
backtracking technique~\citep{Ding:03},
with the augmentation via an additional index recording the signature distance.
A complete description of the new sampler will be provided in a future publication.
The constraint on the signature distance  changes the equilibrium distribution of
structures via eliminating those that are inconsistent with signature
over certain ratio $q$.
Table~\ref{Tab:SnDistESB} shows the means and variances  of the normalized signature  distance for $\Omega^q$.
In particular, we observe that  Boltzmann samples $\Omega_{\text{probe}}$ incorporating the probing data via
pseudo-energies behave similarly as $q$-samples having $q=0.05$.

\section{State-of-the-art experimental approaches}
\label{S:oracle}

Determination of base pairs  is a fundamental and longstanding problem in RNA
biology.
A large variety of experimental approaches  have been developed 
to provide reliable solutions to the problem,
such as X-ray crystallography, nuclear
magnetic resonance (NMR), cryogenic electron microscopy (cryo-EM), 
chemical and enzymatic probing, cross-linking~\citep{Shi:14,Bothe:11,Bai:15,Weeks:15}.
Each method has certain strengths and limitations.
In particular, chemical probing, as one of the most widely accepted experiments,
allows to detect RNA duplexes \textit{in vitro} and \textit{in vivo}, 
and has been combined with high-throughput sequencing 
to facilitate large-scale  analysis on lncRNAs~\citep{Weeks:15}.
Thus, in the following, we focus on determining the queried base pairs via chemical probing.

Chemical probing data is \emph{one-dimensional}, i.e. 
it does not specify base pairing partners.
Thus probing data itself does not directly detect base pairings,
and any structure information can only be \emph{inferred} based on
compatibility with probing data. 
Two strategies of structural inference have been developed, correlation analysis
and mutate-and-map.
\cite{Mustoe:19} introduce 
PAIR-MaP, which utilizes mutational profiling 
as a sequencing approach 
and correlation analysis on profiles. The authors claim that PAIR-MaP provides 
around $0.90$ accuracy of structure modeling 
(on average,  sensitivity $0.96$ and false discovery rate $0.03$).
\cite{Cheng:17}  introduce M2-seq, a mutate-and-map approach combined with next generation
sequencing, which recovers duplexes with a low false discovery rate ($<0.05$).


\section{Structural entropy}
\label{A:3}

\begin{proposition}\label{P:bound}
	Let $\Omega'$ be a sample of size $N$ and $s\in \Omega'$ be a structure having probability
	$p_0$. Then the structural entropy of $\Omega'$ is bounded by 
	\begin{equation*}
	H_{\min}(p_0) \leq H(\Omega') \leq H_{\max} (p_0),
	\end{equation*}
	where
	\begin{align*}
	H_{\min}(p_0) &=-p_0 \log_2 p_0 -(1-p_0)  \log_2 (1-p_0),\\
	H_{\max}(p_0) &=-p_0 \log_2 p_0+(1-p_0)  \log_2 N.
	\end{align*}
\end{proposition}

\begin{proof}
	By construction, the multiplicity of $s$ in $\Omega'$ is given by $p_0^N=\lfloor p_0 N\rfloor$.
	Since the function $-x  \log_2  x$ is for $x>0$ concave, the structural entropy is maximal in
	case of all remaining $N-p_0^N $ structures being distinct, i.e.~each occurs with probability
	$(1-p_0)/(N-p_0^N)=1/N$.
	Therefore
	\begin{align*}
	H_{\max}(p_0) &=-p_0 \log_2 p_0 -\sum_{N-p_0^N}  \frac{1}{N}\log_2  \frac{1}{N}\\
	&=-p_0 \log_2 p_0+(1-p_0)  \log_2 N.
	\end{align*}
	On the other hand, the minimum is achieved 
	when all remaining  structures are the same. Thus
	$	 H_{\min}(p_0)  =-p_0 \log_2 p_0 -(1-p_0)  \log_2 (1-p_0)$. 
\end{proof}

Now we prove Proposition~\ref{C:bound2}.

\begin{proof}[Proof of Proposition~\ref{C:bound2}]
	Let $s_0$ be the structure having the highest probability $p_0$ in $\Omega'$.
	By Proposition~\ref{P:bound},
	we have 
	\begin{equation}\label{Eq:ineq1}
	H_{\min}(p_0) \leq E.
	\end{equation}
	Inspection of the graph of $H_{\min}(p)$ as a function  of $p$, we conclude, that for $E<1$,
	two solutions of the equation $H_{\min}(p)=E$ exist, one being for $f(E)> 0.5$ and the other
	for $g(E) < 0.5$, see Fig.~\ref{F:Hp}. In case of $E=1$, we have the unique solution,
	$f(E)=g(E)=0.5$. Since $H_{\min}(p)$ is monotone over $[0,0.5]$ and $[0.5,1]$,
	inequality~(\ref{Eq:ineq1}) implies 
	$$
	p_0 \geq f(E) \quad \text{\rm  or}\quad    p_0 \leq g(E).
	$$
	We shall proceed by excluding $p_0 \leq g(E)$. A contradiction, suppose that $p_0<0.5$
	and that structures in $\Omega'$ are arranged in descending order according to their
	probabilities  $p_i $ for $i=0,1,\ldots, k$. Since each structure in $\Omega'$  has
	probability smaller than $0.5$, the sample $\Omega'$ contains at least three different
	structures, i.e.  $k\geq 2$. By construction, we have $p_i \leq p_0 <0.5 $.
	Now we consider the following optimization problem
	\begin{align*}
	\min_{p_i} \quad &\sum_{i=0}^k p_i\log_2 p_i \\
	\textrm{s.t.} \quad & \sum_{i=0}^k p_i=1 \\
	&0\leq p_k \leq p_{k-1}\leq \cdots \leq p_0 \leq 0.5.  \\
	\end{align*}
	We inspect that the multivariate function $\sum_{i=0}^k p_i\log_2 p_i$ reaches  its minimum $1$
	only for $ p_0=p_1=0.5$ and $p_i=0$ for $i \geq 2$.
	In the case of $ p_0<0.5$, the minimum cannot be reached and we arrive at some $E>1$, in
	contradiction to our assumption $E\leq 1$. Therefore $p_0 \geq f(E)$ is the only possible scenario,
	i.e., $\Omega'$ contains a distinguished structure with probability at least $f(E)$.
\end{proof}

\begin{figure}
	\centering
	\includegraphics[width=0.5\textwidth]{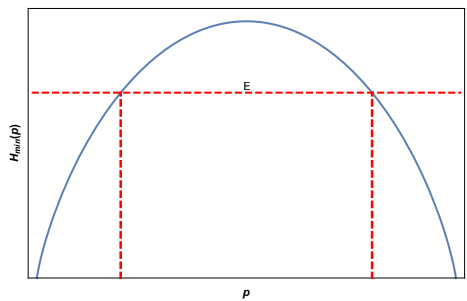}
	\caption
	{\small The graph of $H_{\min}(p)$ as a function  of $p$.
	}
	\label{F:Hp}
\end{figure}


\section{Information theory}
\label{A:4}
As the Boltzmann ensemble is a particular type of discrete probability spaces,
the information-theoretic results on the ensemble trees will be stated in the more general setup.
Let  $\Omega=(\mathcal{S},\mathcal{P}(\mathcal{S}),p)$ be a discrete probability space
consisting of the sample space $\mathcal{S}$, its power set $\mathcal{P}(\mathcal{S})$ as the $\sigma$-algebra
and the probability measure $p$.
The \emph{Shannon entropy} of $\Omega$ is given by 
\begin{equation*}
H(\Omega) = -\sum_{s\in \mathcal{S}} p(s) \log_2 p(s),
\end{equation*}
where the units of $H$ are
in bits.

A \emph{feature} $X$ is a discrete random variable defined on $\Omega$. 
Assume that $X$ has a finite number of values $x_1,x_2,\ldots,x_k$. Set $q_i=\mathbb{P}(X=x_i)$. 
The \emph{Shannon entropy} $H(X)$  of the feature $X$ is given by
\begin{equation*}
H(X) = -\sum_{i} q_i \log_2 q_i.
\end{equation*}

In particular, the values of $X$ define a partition of  $ \mathcal{S}$ 
into disjoint subsets $ \mathcal{S}_i=\{s\in  \mathcal{S}:X(s)=x_i\}$, for $1\leq i \leq k$.
This further induces $k$ spaces $\Omega_i=(\mathcal{S}_i,\mathcal{P}(\mathcal{S}_i),p_i)$,
where the induced distribution is given by
\[
p_i(s)=\frac{p(s)}{q_i} \quad \text{ for } s\in  \mathcal{S}_i,
\]
and $q_i$ denotes the probability of $X$ having value $x_i$ and is given by
\begin{equation*}
q_i =\mathbb{P}(X=x_i)= \sum_{s\in  \mathcal{S}_i} p(s).
\end{equation*}
Let $H(\Omega|X) $ denote the \emph{conditional entropy} of $\Omega$ given the value of feature $X$. 
The entropy $H(\Omega|X) $ gives the expected value of the entropies of the conditional distributions on $\Omega$, 
averaged over the conditioning feature $X$ and  can be computed by 
\begin{equation*}
H(\Omega|X)= \sum_i q_i H(\Omega_i).
\end{equation*}

Then the \emph{entropy reduction} 
$R(\Omega,X)$ of $\Omega$ for feature $X$ 
is the difference between the \textit{a priori} Shannon entropy $H(\Omega)$ and the conditional entropy $H(\Omega|X) $, i.e. 
\begin{equation*}
R(\Omega,X)= H(\Omega)-H(\Omega|X). 
\end{equation*}
The entropy reduction indicates the change on average in information entropy from a prior state to a state that takes some information as given.

Now we prove Propositions~\ref{P:1} and~\ref{P:subspace}. 

\begin{proof}[Proof of Proposition~\ref{P:1}]
	\begin{align*}
	H(\Omega|X)&= \sum_i q_i H(\Omega_i)\\
	&= - \sum_i q_i \sum_{s\in \mathcal{S}_i} p_i(s) \log_2 p_i(s)\\
	&=- \sum_i q_i \sum_{s\in \mathcal{S}_i} \frac{p(s)}{q_i} \log_2\frac{p(s)}{q_i}\\
	&=- \sum_i  \sum_{s\in \mathcal{S}_i} p(s) ( \log_2 p(s)- \log_2 q_i)\\
	&= - \sum_i  \sum_{s\in \mathcal{S}_i} p(s)  \log_2 p(s)+
	\sum_i \log_2 q_i \sum_{s\in \mathcal{S}_i} p(s) \\
	&=-\sum_{s\in \mathcal{S}} p(s) \log_2 p(s)+\sum_{i} q_i \log_2 q_i\\
	&=H(\Omega) -H(X).
	\end{align*}	
	Therefore eq.~(\ref{Eq:IR}) follows.
\end{proof}

\begin{proof} [Proof of Proposition~\ref{P:subspace}]
	By definition,
	$$	\mathbb{P}(\Omega_1^{i,j})=\sum_{s \in \Omega_1^{i,j}} p(s)
	=\mathbb{P}(X_{i,j}(s)=1)
	=p_{i,j}.$$
	Similarly, we have $\mathbb{P}(\Omega_0^{i,j})=1-p_{i,j}$. Thus $|\mathbb{P}(\Omega_0^{i,j}) -\mathbb{P}(\Omega_1^{i,j}) |=|1-2 p_{i,j}|$
	is strictly decreasing on $p_{i,j}  \in [0,1/2]$ and strictly increasing on $[1/2,1]$.
	Meanwhile,  the function $H(X_{i,j})=-p_{i,j}  \log_2 p_{i,j}  -(1-p_{i,j} )  \log_2 (1-p_{i,j} )$ 
	is strictly increasing on $p_{i,j}  \in [0,1/2]$ and   symmetric with respect to $p_{i,j} =1/2$.
	Therefore, $|\mathbb{P}(\Omega_0^{i,j}) -\mathbb{P}(\Omega_1^{i,j}) |$ reaches its minimum 
	when  $H(X_{i,j})$ has the maximum value, that is,  $X_{i_0,j_0}$.
	
	Assertion (2) follows directly from   Proposition~\ref{P:1}.
\end{proof}

Given two features $X_1$ and $X_2$, we can  partition $\Omega$ either   first by $X_1$ and  subsequently  by  $X_2$,
or  first by $X_2$ and then  by  $X_1$, or just by  a pair of features $(X_1,X_2)$.  
In the following, we will show that all three approaches provide the same entropy reduction of $\Omega$.

Before the proof, we define some notations. 
The joint probability distribution of  a pair of features $(X_1,X_2)$ is given by 
$q_{i_1,i_2}= \mathbb{P}(X_1=x^{(1)}_{i_1},X_2=x^{(2)}_{i_2})$, 
and the marginal probability distributions are given by
$q^{(1)}_{i_1}=\mathbb{P}(X_1=x^{(1)}_{i_1})$ and $q^{(2)}_{i_2}=\mathbb{P}(X_2=x^{(2)}_{i_2})$.
Clearly, $\sum_{i_1}q_{i_1,i_2}=q^{(2)}_{i_2} $ and $\sum_{i_2}q_{i_1,i_2}=q^{(1)}_{i_1} $.
The \emph{joint entropy} $H (X_1,X_2) $ of a pair  $(X_1,X_2)$ is defined as
\begin{equation*}
H (X_1,X_2)= -\sum_{i_1} \sum_{i_2}  q_{i_1,i_2}  \log_2 q_{i_1,i_2}.
\end{equation*}

The \emph{conditional entropy} $H(X_2|X_1)$ of a feature $X_2$ given $X_1$ is 
defined as the expected value of  the entropies of  the conditional distributions $X_2$, 
averaged over the conditioning feature $X_1$, i.e. 
\begin{equation*}
H(X_2|X_1)= \sum_{i_1}  \mathbb{P}(X_1=x^{(1)}_{i_1})  H(X_2|X_1=x^{(1)}_{i_1}).
\end{equation*}

\begin{proposition}[Chain rule,~\cite{Cover:06}]\label{P:chain}
	\begin{equation}\label{Eq:chain}
	H(X_1,X_2)=H(X_1)+ H(X_2|X_1).
	\end{equation}
\end{proposition}

\begin{proposition}\label{P:2}
	Let  $R(\Omega,X_1,X_2)$ denote   the \emph{entropy reduction} of $\Omega$ first by the feature $X_1$ and then by the feature $X_2$, 
	and $R(\Omega,(X_1,X_2))$ denote   the \emph{entropy reduction} of $\Omega$  by a pair of features $(X_1,X_2)$. Then
	\begin{equation}\label{Eq:IR2}
	R(\Omega,X_1,X_2) = R(\Omega,(X_1,X_2)).
	\end{equation}
\end{proposition}

\begin{proof}
	By Proposition~\ref{P:1}, we have 
	\begin{align*}
	R(\Omega,X_1) & = H(X_1),\\
	R(\Omega,(X_1,X_2))&= H(X_1,X_2).
	\end{align*}
	Let  $\Omega_{i_1}$ denote the spaces obtained by partitioning $\Omega$ via $X_1$, i.e. 
	$\Omega_{i_1}=(\mathcal{S}_{i_1},\mathcal{P}(\mathcal{S}_{i_1}),p_{i_1})$,
	where $ \mathcal{S}_{i_1}=\{s\in  \mathcal{S}:X_1(s)=x^{(1)}_{i_1}\}$, 
	and
	\[
	p_{i_1}(s)=\frac{p(s)}{q^{(1)}_{i_1}}, \quad \text{ for } s\in  \mathcal{S}_{i_1},
	\]
	where $q^{(1)}_{i_1}=\mathbb{P}(X_1=x^{(1)}_{i_1})$.
	Then the space $\Omega_{i_1}$  is further partitioned into $\Omega_{i_1,i_2}$ via $X_2$.
	That is, $\Omega_{i_1,i_2}=(\mathcal{S}_{i_1,i_2},\mathcal{P}(\mathcal{S}_{i_1,i_2}),p_{i_1,i_2})$,
	where $ \mathcal{S}_{i_1,i_2}=\{s\in  \mathcal{S}_{i_1}:X_2(s)=x^{(2)}_{i_2}\}$, 
	and
	\[
	p_{i_1,i_2}(s)=\frac{p_{i_1}(s)}{\mathbb{P}(X_2=x^{(2)}_{i_2} |X_1=x^{(1)}_{i_1})} =  \frac{\frac{p(s)}{q^{(1)}_{i_1}}  }{ \frac{q_{i_1,i_2}}{q^{(1)}_{i_1}}}= \frac{p(s)}{q_{i_1,i_2}},   \quad \text{ for } s\in  \mathcal{S}_{i_1,i_2}.
	\]
	The entropy reduction $R(\Omega,X_1,X_2)$ is given by the difference between 
	the \textit{a priori} Shannon entropy $H(\Omega)$ and the conditional entropy $H((\Omega|X_1)|X_2)$,
	which is the expected value of the entropies of 
	$\Omega_{i_1,i_2}$, weighted by the probability 
	$\mathbb{P}(s\in  \mathcal{S}_{i_1,i_2}) =\mathbb{P}(X_2=x^{(2)}_{i_2}, X_1=x^{(1)}_{i_1}) = q_{i_1,i_2}$.
In view of Proposition~\ref{P:1}, we derive
	\begin{align*}
	R(\Omega,X_1,X_2)& = H(\Omega)-H((\Omega|X_1)|X_2)\\
	&= H(\Omega) -  \sum_{i_1,i_2} \mathbb{P}(s\in  \mathcal{S}_{i_1,i_2}) H(\Omega_{i_1,i_2})\\
	&= H(\Omega) +\sum_{i_1,i_2}  \mathbb{P}(s\in  \mathcal{S}_{i_1,i_2})  \sum_{s\in  \mathcal{S}_{i_1,i_2} } p_{i_1,i_2}(s) \log_2 p_{i_1,i_2}(s) \\
	&= H(\Omega) +\sum_{i_1,i_2}q_{i_1,i_2} \sum_{s\in  \mathcal{S}_{i_1,i_2} } \frac{p(s)}{q_{i_1,i_2}} \log_2 \frac{p(s)}{q_{i_1,i_2}} \\
	&=H(\Omega) +  \sum_{i_1,i_2} \sum_{s\in  \mathcal{S}_{i_1,i_2} } p(s) \log_2 p(s)- \sum_{i_1,i_2} \sum_{s\in  \mathcal{S}_{i_1,i_2} }  p(s) \log_2 q_{i_1,i_2}\\
	&=H(\Omega) +\sum_{s\in  \mathcal{S} } p(s) \log_2 p(s) -   \sum_{i_1,i_2}  \log_2 q_{i_1,i_2} \sum_{s\in  \mathcal{S}_{i_1,i_2} }  p(s)\\
	&=H(\Omega) -H(\Omega) -   \sum_{i_1,i_2} q_{i_1,i_2}  \log_2 q_{i_1,i_2}\\
	&= H(X_1,X_2)\\
	&= 	R(\Omega,(X_1,X_2)).
	\end{align*}
	Eq.~(\ref{Eq:IR2}) follows.
\end{proof}

The maximum entropy of an arbitrary feature is achieved when all  its outcomes   occur with equal probability, 
and this maximum value is proportional to the logarithm of the number of  possible outcomes  to the base $2$.
Thus Proposition~\ref{P:1}  implies that the more possible outcomes  a feature has, 
the higher  entropy reduction it could possibly lead to.

Meanwhile,  a feature with an arbitrary number of outcomes can be viewed as a combination of  \emph{binary features}, 
the ones with two possible outcomes.
Even though the entropy of the combination of two features is greater than  each of them,
Proposition~\ref{P:2} shows that partitioning the space  subsequently  by two features has the same entropy reduction
as partitioning by their combination. 
Therefore, instead of considering features with outcomes as many as possible, we 
focus on binary features.


\section{Query repeats}
\label{A:error}
Here we assess the improvement of the error rates  by repeating the same query  twice.
Let $Y$ (or $N$) denote the event of the queried base pair  existing (or not) in the target structure. 
Let $y$ (or $n$) denote the event of the experiment confirming (or rejecting) the base pair.
Let $nn$ denote the event of two independent experiments both rejecting the base pair. 
Similarly, we have $yy$ and $y n$.
Utilizing the same sequences and structures as described in  Fig.~\ref{F:shapeM},
we estimate the conditional probabilities $	\mathbb{P}(n|N)\approx 0.993$ and $	\mathbb{P}(n|Y)\approx 0.055$.
The \emph{prior probability} $\mathbb{P}(Y)$ can be computed via the expected number  $l_1$ 
of confirmed queried base pairs on the path,
divided by the number of queries in each sample. 
Fig.~\ref{F:YesDistribution} displays the distribution of $l_1$ having mean around $5$.
Thus we adopt 
$\mathbb{P}(Y)= \mathbb{P}(N)=0.5$.
By Bayes' theorem, we calculate the \emph{posterior} 
\[
\mathbb{P}(N|nn)=\frac{\mathbb{P}(nn|N) \mathbb{P}(N)}{\mathbb{P}(nn)}=\frac{\mathbb{P}(n|N)^2 \mathbb{P}(N)}{\mathbb{P}(n|N)^2 \mathbb{P}(N) +\mathbb{P}(n|Y)^2\mathbb{P}(Y)},
\]
where $\mathbb{P}(nn)=\mathbb{P}(nn|N) \mathbb{P}(N) + \mathbb{P}(nn|Y) \mathbb{P}(Y)$.
Since two  experiments can be assumed to conditionally independent 
given $Y$ and also given $N$,
we have $\mathbb{P}(nn|N) =\mathbb{P}(n|N)^2$ and $\mathbb{P}(nn|Y) =\mathbb{P}(n|Y)^2$.
Similarly, we compute $\mathbb{P}(Y|nn)$, $\mathbb{P}(Y|yy)$ and $\mathbb{P}(Y|yn)$ etc, see
Table~\ref{Tab:twice}.
It demonstrates that,  if we get the same  answer to the query twice,
the error rates would become significantly smaller, for example,
$e_0^{[2]} =0.003$ and $e_1^{[2]} =0.00005$.
In the case of mixed answers $ny$ or $yn$,
its probability $\mathbb{P}(ny)=0.0292$,
i.e., it rarely happens. We would  recommend a third experiment 
and take the majority of three answers
when getting two mixed answers.

In principle, we can extend to reducing the error rates by repeating the same query $k$ times.
The above Bayesian argument is then generalized to sequential updating on  the error rates from $e_0$ to 
$e_0^{[k]} $. We can show that $e_0^{[k]}$ and $e_1^{[k]} $ approach to $0$,
as $k$ grows to infinity. In this case, the reliability of the leaf space $\mathbb{P}(s\in \Omega_{11}) $
is $1$, i.e.  the leaf always contain the target.
The fidelity of 
the distinguished structure $\mathbb{P}(s^{*}=s ) $ increases from $70\%$ to $94\%$
for sequences of length $300$.
To sum up, asking the same query a constant number of times  
significantly improves the fidelity of the leaf  and the distinguished structure.
\begin{table}[htbp]
	\caption{ The posterior probabilities after two experiments. We use the same sequences and structures as described in Fig.~\ref{F:shapeM}. } \label{Tab:twice}
	\begin{tabular}{ccc}
		\hline\noalign{\smallskip}
		\small 	Outcome of two experiments  &\small $Y$  &\small  $N$ \\
		\noalign{\smallskip}\hline\noalign{\smallskip}
		\small 	$nn $  &\small  $\mathbb{P}(Y|nn)=0.003$ &\small $\mathbb{P}(N|nn)=0.997$\\
		\small   $yy$ 	&\small $\mathbb{P}(Y|yy)=0.99995$ &\small $\mathbb{P}(N|yy)=0.00005$
		\\   	    
		\small   $ny$  or $yn $	&\small $\mathbb{P}(Y|ny)=0.881$ &\small $\mathbb{P}(N|ny)=0.119$
		\\   	    
		\noalign{\smallskip}\hline\noalign{\smallskip}	            
	\end{tabular}
\end{table}


\section{A new fragmentation}
\label{A:5}
Here we  present a novel fragmentation process, guided by the base-pair queries of  the ensemble tree
inferred from the restricted Boltzmann sample incorporating chemical probing.
Given the maximum entropy base pair, $(i,j)$, extraction splits the sequence into two fragments,
one being the extracted fragment $\mathbf{x}_{i,j}$ and the other, $\bar{\mathbf{x}}_{i,j}$,
i.e. $ \xi_{i,j}(\mathbf{x})=(\mathbf{x}_{i,j},\bar{\mathbf{x}}_{i,j})$.
We perform probing experiments on these two segments, and obtain the reactive probabilities
$\mathbf{q}_{i,j}$ and  $\bar{\mathbf{q}}_{i,j}$, respectively.
Let  $\mathbf{q}$ be the reactive probability for the entire sequence, and $\mathbf{q}'$ be
the embedding  of $\mathbf{q}_{i,j}$ into $\bar{\mathbf{q}}_{i,j}$, 
i.e. $\mathbf{q}'= \epsilon_{i,j} (\mathbf{q}_{i,j}, \bar{\mathbf{q}}_{i,j})$.
As shown in Section~\ref{S:modular},
if the Hamming distance $d(\mathbf{q},\mathbf{q}') $ is smaller than  threshold $\theta$, then the probing profiles are similar, i.e.~two bases $i$
and $j$ are paired. Otherwise, they are unpaired in the target structure.

The fragmentation procedure can be summarized as follows:

\begin{enumerate}
	\item a probing experiment for the entire sequence is performed and the reactive
	probability $\mathbf{q}$ is obtained,
	\item a Boltzmann sample $\Omega_{\text{probe}}$ of $N$ structures, consistent 
	with the probing data $\mathbf{q}$ is computed, 
	\item the ensemble tree $T(\Omega)$ containing the sub-spaces $\Omega_{\mathbf{k}}$ 
	and the corresponding maximum entropy base pairs $X_{i_{\mathbf{k}},j_{\mathbf{k}}}$ is constructed,
	\item starting with $\Omega$ we recursively answer the queries, determining thereby a path
	through the ensemble tree from the root to a leaf.
	\item once in a leaf, Proposition~\ref{C:bound2} guarantees the existence of a distinctive structure
	which we stipulate to be the target structure.
\end{enumerate}

\bibliographystyle{elsarticle-harv} 
\bibliography{ref}   

\end{document}